\documentclass{article}
\usepackage{ijcai26}
\usepackage{float}

\usepackage{times}
\usepackage{soul}
\usepackage{url}
\usepackage[hidelinks]{hyperref}
\usepackage[utf8]{inputenc}
\usepackage[small]{caption}
\usepackage{graphicx}
\usepackage{amsmath}
\usepackage{amsthm}
\usepackage{booktabs}
\usepackage{algorithm}
\usepackage{algorithmic}
\usepackage[switch]{lineno}

\usepackage{subcaption}
\usepackage{url}
\usepackage{graphicx}

\usepackage{amssymb,amsmath,amsfonts}
\usepackage{mathtools}
\usepackage{ifthen,version}
\usepackage{tikz}
\usepackage{todonotes}
\usetikzlibrary{automata,positioning,decorations.markings,arrows,intersections,calc,shapes,patterns}
\usepackage{cleveref}
\usepackage{thm-restate}

\pgfmathsetmacro{\arrLen}{1}        
\pgfmathsetlengthmacro{\HorizVarDistnc}{2cm}

\usepackage{pgfplots}
\pgfplotsset{compat=1.16}

\usepackage{nameref}

\usepackage{pgfplots, pgfplotstable}
\pgfplotsset{compat=1.9}

\usepackage{amsthm}
\newtheorem{example}{Example}
\newtheorem{theorem}{Theorem}
\newtheorem{definition}{Definition}

\newtheorem{lemma}{Lemma}
\newtheorem{remark}{Remark}
\newtheorem{observation}{Observation}
\newtheorem{assumption}{Assumption}

\usepackage{algorithm}
\usepackage{algorithmic}

\colorlet{darkgreen}{green!40!black}
\colorlet{darkblue}{blue!60!black}
\colorlet{darkred}{red!50!black}
\colorlet{safecellcolor}{yellow!5}
\colorlet{goodcellcolor}{green!10}
\colorlet{badcellcolor}{blue!10}

\PassOptionsToPackage{usenames,dvipsnames}{xcolor}

\definecolor{dkgreen}{rgb}{0,0.6,0}
\definecolor{gray}{rgb}{0.5,0.5,0.5}
\definecolor{mauve}{rgb}{0.58,0,0.82}

\tikzset{
  >=latex,node distance=2cm,on grid,auto, initial text=,
  box state/.style={draw,rectangle,minimum size=8mm,rounded corners},
  prob state/.style={draw,very thick,shape=circle,darkblue,minimum size=3mm,inner sep=0mm},
  every loop/.style={shorten >=0pt},
  accepting state/.style={double distance=1.2pt, outer sep = 0.6pt+\pgflinewidth},
  accepting dot/.style={above=-2.5pt,circle,fill,darkgreen,inner sep=2pt,radius=1pt},
  loop above/.append style={every loop/.append style={out=120, in=60, looseness=6}},
  loop below/.append style={every loop/.append style={out=300, in=240, looseness=6}},
  loop left/.append style={every loop/.append style={out=210, in=150, looseness=6}},
  loop right/.append style={every loop/.append style={out=30, in=330, looseness=6}}
}

\newcommand{\AAND}{\mathsf{AND}}
\newcommand{\OOR}{\mathsf{OR}}

\newcommand{\mdp}{\mathcal M}
\newcommand{\sw}{\mathsf{SW}}
\newcommand{\pto}{\xrightharpoondown{}}
\newcommand{\set}[1]{\left\{ #1 \right\}}

\newcommand{\seq}[1]{\langle #1 \rangle}
\newcommand{\FRuns}{\mathsf{FRuns}}
\newcommand{\Runs}{\mathsf{Runs}}

\newcommand{\EDisct}{\mathcal{D}}

\newcommand{\run}{\rho}

\newcommand{\Ff}{\mathcal{F}}

\newcommand{\eE}{\mathbb E}

\newcommand{\Real}{\mathbb R}

\newcommand{\DIST}{{\cal D}}

\DeclareMathOperator{\supp}{\mathit{supp}}
\DeclareMathOperator{\last}{\mathsf{last}}

\iftrue
\newcommand{\soum}[1]{\todo[fancyline,size=footnotesize,color=cyan!20,caption={SP}]{\textbf{SP:} #1}}
\newcommand{\soumline}[1]{\todo[inline,size=footnotesize,color=cyan!20,caption={SP}]{\textbf{SP:} #1}}

\newcommand{\ashu}[1]{\todo[fancyline,size=footnotesize,color=orange!20,caption={SP}]{\textbf{AT:} #1}}
\newcommand{\ashuline}[1]{\todo[inline,size=footnotesize,color=orange!20,caption={SP}]{\textbf{AT:} #1}}

\else
\newcommand{\soum}[1]{}
\newcommand{\soumline}[1]{}
\newcommand{\ashu}[1]{}
\newcommand{\ashuline}[1]{}
\fi

\title{Social Welfare under Heterogeneous Time Preferences}

\author{
Sarvin Bahmani$^1$\and
Soumyajit Paul$^1$\and
Sven Schewe$^1$\and
Shadi Tasdighi Kalat$^2$\and
Ashutosh Trivedi$^{1,2}$\\
\affiliations
$^1$University of Liverpool, UK \\
$^2$University of Colorado Boulder, USA\\
\emails
\{R.Bahmani, Soumyajit.Paul, Sven.Schewe\}@liverpool.ac.uk, \\\{Shadi.TasdighiKalat, Ashutosh.Trivedi\}@colorado.edu
}

\begin{document}

\maketitle

\begin{abstract}
In several socioeconomic-critical decision-making settings, such as fair resource allocation, climate policy, or AI alignment, multiple principals interact within a common arena. While it is well established that these principals may have differing preferences, decision-making under heterogeneous time preferences remains relatively unexplored. In particular, principals may weigh future outcomes differently and may derive distinct utilities from the same decisions.
Motivated by such scenarios, we introduce the notion of heterogeneous time preferences in MDPs, where multiple principals possess distinct reward functions and apply different discount factors to future rewards. 
To compute meaningful decisions in such settings, an AI agent must rely on a notion of optimality that accounts for the preferences of all principals. 

We adopt a utilitarian notion of social welfare, defined as the sum of utilities accrued to all principals, and study the synthesis of agent strategies that maximise this welfare. Under heterogeneous time preferences, we show that optimal strategies are no longer positional, even when all principals receive identical rewards. Nevertheless, optimal strategies remain structurally simple: they can be realized as pure finite-memory counting strategies, require only polynomial memory in the system size, and can be synthesized in polynomial time.
On the other hand, we show that deciding threshold questions for optimal positional strategies is NP-hard, exposing a poor trade-off: insisting on positional simplicity neither makes synthesis tractable nor preserves social welfare.

\end{abstract}


\section{Introduction}
\label{sec:intro}
Markov decision processes (MDPs) are canonical models of decision-making under uncertainty and form the backbone of several mature and widely deployed technologies, including reinforcement learning~\cite{Sutton18}, optimal control~\cite{Put94}, and economics and game theory~\cite{basar1999dynamic,filar1997competitive}.
In settings where a single decision-making agent acts on behalf of multiple principals, the preferences of each principal are encoded by a distinct reward function that quantifies the instantaneous utility they associate with each decision.
This work addresses the problem of synthesizing strategies for the agent that maximise a utilitarian notion of \emph{social welfare}, defined as the aggregate discounted payoff across all principals.

\paragraph{Time Preference and Discounting.}
As an agent interacts with a system, it accrues infinite sequences of rewards on behalf of multiple principals.
These rewards are typically aggregated using geometric discounting to evaluate the overall utility of a strategy.
The discount factor plays a dual role. First, it imposes a notion of effective finiteness on the expected trajectory length. Second, it endows the optimisation process with desirable mathematical properties, such as value contraction~\cite{Shapley53}.
The discount factor also admits a behavioural interpretation~\cite{tasdighi2024your,TasdighiKalat2026} rooted in human decision-making and captures the principle of \emph{time preference}, namely that a dollar today is valued more than a dollar tomorrow.

\paragraph{Heterogeneous Time Preferences.}
In many real-world settings, principals do not share identical attitudes toward delayed rewards.
While discounting is commonly used to aggregate long-run payoffs, the choice of discount factor reflects an underlying time preference and may vary across principals due to structural incentives, institutional roles, or strategic priorities.
Consequently, assuming a common discount factor can be overly restrictive and may obscure natural asymmetries in how principals value future outcomes.

This motivates the study of MDPs in which multiple principals, each with their own reward function and discount factor, delegate decision-making to a single agent constrained to follow a common strategy.
In such settings, the agent must reconcile heterogeneous time preferences when optimizing behaviour.
Our objective is to synthesize strategies that resolve this tension in service of the common good by maximizing a utilitarian notion of \emph{social welfare}, defined as the aggregate discounted payoff across all principals.

\input{Figures/fig-example-payoffs}

\paragraph{The Need for Memory.}
In the well-studied setting where all principals share the same discount factor, optimal strategies are known to be positional.
In contrast, we show that memory becomes necessary to maximise social welfare even in simple settings involving only two principals who receive identical rewards at every decision point.

\begin{example}[{\bf Memory is necessary under heterogeneous discounting}]
\label{ex:needmem}
    Consider a scenario in which two principals, Alice and Bob, jointly operate a hotel and must decide whether to update and expand it. If they continue operating the hotel without investment (action $a$ from state $s_0$), they each earn \$3 million per year. If they choose to update and expand (action $b$ from $s_0$), they incur a one-time cost of \$1 million each and transition to state $s_1$, where they earn \$6 million annually thereafter; see~\Cref{fig:exp-same-payoff}. Although both principals receive the same rewards, they differ in their discount factors: Alice discounts the future at $\alpha = \frac{2}{3}$, while Bob discounts more steeply at $\beta = \frac{1}{3}$.

    The space of strategies consists of either remaining in $s_0$ forever, denoted $a^\omega$, or staying in $s_0$ for $k$ steps and then transitioning to $s_1$, denoted $a^k b^\omega$. For each principal with discount factor $\lambda \in \{\alpha, \beta\}$, the discounted payoff for the strategy $a^\omega$ is 
    $\frac{3}{1 - \lambda}.$
    For the strategy $a^k b^\omega$ is:
    \[
   3\cdot\frac{1 - \lambda^k}{1 - \lambda} - \lambda^k + 6\cdot \frac{\lambda^{k+1}}{1 - \lambda} = \frac{3 - 4\lambda^k + 7 \lambda^{k+1}}{1 - \lambda}.
    \]

If the agent were to maximise Alice’s payoff alone, the optimal strategy would be to invest and transition to $s_1$; in contrast, if optimizing solely for Bob, given his steeper discounting, the agent would prefer to remain in $s_0$ indefinitely. 
However, neither of these individually optimal strategies maximises social welfare. 

The socially optimal strategy is to remain in $s_0$ for exactly two steps before transitioning to $s_1$, i.e., to follow the strategy $a^2 b^\omega$. Although this compromise is suboptimal for each principal's discounted payoff in isolation, it yields a higher social welfare. Bob’s gain from deferring the transition outweighs Alice’s marginal loss, resulting in a net improvement in aggregate welfare. In contrast, delaying the transition beyond two steps leads to a rapid decline in Alice’s utility that surpasses any marginal gain Bob receives, thereby decreasing social welfare. \Cref{fig:plot-same-payoff} illustrates the individual discounted payoffs of Alice and Bob, as well as the resulting social welfare, for the different strategies.

The optimal number of steps to remain in $s_0$ before transitioning to $s_1$ depends on the specific discount factors of the two principals. \Cref{fig:strategy-map} shows how the socially optimal strategy varies over the space of discount factor pairs. Each region in the map corresponds to a fixed number of steps the agent must wait in $s_0$ before executing the transition.
\end{example}

\paragraph{Contributions.}
We study a setting in which a single agent acts on behalf of multiple principals to optimise a utilitarian notion of \emph{social welfare}, defined as the total discounted payoff accrued to all principals. The framework explicitly accommodates heterogeneous time preferences by allowing each principal to have an individual discount factor that contributes to the aggregated objective.

We first show that \emph{memoryless} strategies are not sufficient for achieving optimal social welfare. Moreover, even when restricting attention to \emph{positional} (memoryless) strategies, optimality may require \emph{randomization}. For example, in the hotel investment scenario (Example~\ref{ex:needmem}), a randomized positional strategy that selects action $a$ in state $s_0$ with probability $1/4$ achieves an expected social welfare of $13.\overline{6}$, exceeding the welfare attained by any pure positional strategy.

We then study the computational complexity of optimising over such simpler strategy classes. We show that determining whether there exists a \emph{stationary} pure or mixed strategy that achieves a given social-welfare threshold is NP-hard.

As our main contribution, we show that socially optimal strategies can always be chosen from a class of \emph{pure finite-memory counting strategies}. Furthermore, under mild assumptions on the spacing of discount factors,\footnote{Specifically, this holds when the quantity $\frac{1}{\alpha/\beta - 1}$ for two discount factors $\alpha > \beta$ is polynomially bounded in the input size, as is the case when denominators are encoded in unary or discount factors are drawn from a finite set.} such strategies can be synthesized efficiently, in polynomial time.
Finally, as a proof of concept, we implement our synthesis algorithm and demonstrate empirically that it scales well in practice.

\paragraph{Related Work.}
Asymmetric discounting has been extensively studied in the context of repeated games, where it enables cooperative outcomes that are otherwise unsustainable under symmetric time preferences. 
 Lehrer and Pauzner~\shortcite{Lehrer-Pauzner-1999} showed that heterogeneity in discount factors can expand the equilibrium payoff set, as more patient players are able to delay gratification to support cooperation—effectively subsidizing less patient players early on. 
 Dasgupta and Ghosh~\shortcite{Dasgupta-Ghosh} further analysed how time-preference asymmetry reshapes incentives and enhances the stability of cooperative behaviour. 
These insights have been extended to multi-agent settings. In particular, Gu{'e}ron et al.~\shortcite{gueron2011folk} and Chen and Takashi~\shortcite{chen2012folk} show that in repeated games with heterogeneous discount factors, if all discount factors are sufficiently close to one, then any strictly individually rational payoff profile can arise as the outcome of a subgame-perfect equilibrium.

However, these works focus on \emph{stateless} repeated games, where each round is structurally identical and independent of past actions, aside from discounting, and where players optimise only their own discounted payoffs.
In contrast, our work introduces time-preference asymmetry into \emph{stateful} stochastic games, where agents (acting on behalf of multiple principals) influence state transitions over time. In this setting, discounting interacts with system dynamics, leading to qualitatively new forms of strategic complexity. To the best of our knowledge, this is the first formal study of cooperative stochastic games with heterogeneous time preferences.

Our setting is fundamentally different from state-dependent discounting in MDPs and stochastic games~\cite{Gan_Hennes_Majumdar_Mandal_Radanovic_2023,chatterjee-LPAR-13}. There, discount factors may depend on the state, but are uniform across agents at each time step.
Pitis~\shortcite{PitisNEURIPS2023} adopts an axiomatic perspective and shows that aggregating objectives with different time preferences cannot be represented by a Markovian reward function. Their analysis considers multiple discount factors for the same agent. In contrast, we study a setting with multiple principals, each endowed with its own discount factor, and fix a utilitarian social-welfare objective. We then analyse the induced optimization problem in stateful stochastic games.

Multi-objective optimisation is a closely related line of work, in which outcomes are evaluated with respect to multiple reward criteria and typical questions involve the synthesis of Pareto-optimal strategies or the comparison of outcomes under partial orders such as lexicographic preferences~\cite{Chatterjee-et-al-STAC06,chatterjee-LPAR-13}.
By contrast, the focus here is not on optimizing multiple objectives per se, but on decision-making on behalf of multiple principals whose discounted utilities are aggregated into a single utilitarian notion of social welfare.

\section{Preliminaries}
\label{sec:prelim}

A \emph{probability distribution} over a finite set $S$ is a function $d \colon S \to [0, 1]$ such that $\sum_{s \in S} d(s) = 1$. Let $\DIST(S)$ denote the set of all probability distributions over $S$.

\begin{definition}
A \emph{Markov Decision Process} (MDP) $\mdp$ is a tuple
$(S, A, T)$, where $S$ is a finite set of states, $A$ is a finite set of actions, and 
$T \colon S \times A \pto \DIST(S)$ is a \emph{probabilistic transition function}.
For $s \in S$, let $A(s)$ denote the set of actions available at $s$.  
For $s, s' \in S$ and $a \in A(s)$, we write $p(s'\mid s, a)$ to denote $T(s, a)(s')$.
\end{definition}

We measure the size of $\mdp$ by the size of its representation, where the values may be given in binary.

A {\it run} $\run$ of MDP $\mdp$ is an
$\omega$-word $\seq{s_0, a_0, s_1, a_1, s_2 \ldots} \in S \times (A \times S)^\omega$ such that $p(s_{j+1} | s_{j}, a_{j}) > 0$ for all $j \geq 0$.  
A finite run is a finite such sequence, that is, a word in
$ S \times (A \times S)^*$. 
Let $\Runs^\mdp (resp.~ \FRuns^\mdp)$ denote the set of runs ($resp$.\ finite runs) of the $\mdp$ and $\Runs{}^\mdp(s) (resp.\ \FRuns{}^\mdp(s))$ for the set of runs ($resp$.\ finite runs) of the $\mdp$ starting from state~$s$. 
We write $\last(\run)$ for the last state of a finite run $\run$.
A strategy of the agent in MDP $\mdp$ is a partial function $\sigma : \FRuns^\mdp \pto \DIST(A)$, such that $\supp(\sigma(\run)) \subseteq A(\last(\run))$ and is only defined for those runs $\run$ with all prefixes conforming to $\sigma$, i.e. $\forall $ prefix $\run' \in \FRuns^\mdp$ of $\run$, for some $a' \in supp(\sigma(\run'))$ and some $s' \in S$, $\run'(a',s')$ is in $\FRuns$ and also a prefix of $\run$.
Let $\Sigma$ be the set of all strategies of MDP $\mdp$. We also consider the following special classes of strategies: 
\begin{itemize}
    \item \emph{Pure:} $\sigma(\run)$ is always a point distribution wherever defined; a non-pure strategy is called a \emph{mixed} strategy.
    \item \emph{Stationary:} action depends only on the current state, i.e. $\last(\run) = \last(\run') \implies \sigma(\run) = \sigma(\run')$ 
    \item \emph{Positional:} both pure and stationary.
    \item \emph{Counting:} depends only on current state and run length.
\end{itemize}

The behaviour of an MDP $\mdp$, is defined on a probability space $(\Runs^\mdp(s), \Ff_{\Runs^\mdp(s)}, \Pr^\mdp(s))$ over the set of infinite runs $\Runs^\mdp(s)$ starting from state $s$ : the sigma-algebra is initially defined over cylinder sets corresponding to finite runs and is extended to a full probability measure over infinite runs using the Ionescu-Tulcea extension theorem.
Given a random variable $f \colon \Runs^\mdp \to \Real$ over the infinite runs of $\mdp$, we denote by
$\eE^\mdp(s) \set{f}$ the expectation of $f$ w.r.t.\ the probability space
$(\Runs^\mdp(s), \Ff_{\Runs^\mdp(s)}, \Pr^\mdp(s))$.

\subsection{Asymmetrically-Discounted MDPs}
As discussed earlier, heterogeneous time preferences arise when different principals evaluate delayed rewards using distinct discount rates. 
To formally capture this setting, we define \emph{asymmetrically-discounted MDPs} as a generalization of standard MDPs in which each principal is associated with their own reward function and discount factor. This model allows us to represent a single decision-making agent acting on behalf of multiple principals who differ in both their valuation of outcomes and their patience over time.

\begin{definition}
An \emph{asymmetrically-discounted MDP} is a tuple $(S, A, T, P, R, \Lambda)$, where:
\begin{itemize}
\item $(S, A, T)$ is a Markov decision process;
\item $P = \set{0, 1, \ldots, n{-}1}$ is a set of $n$ principals (players);
\item $R \colon S \times A \times P \to \mathbb{Q}$ is a reward function specifying the immediate reward received by each principal for each valid state-action pair;
\item $\Lambda = \set{\lambda_p \in (0,1) \cap \mathbb{Q} \mid p \in P}$ assigns each principal a rational discount factor.
\end{itemize}
Without loss of generality,\footnote{For the purpose of social welfare, players with equal discount factors may be merged by summing their rewards.} we assume the discount factors are ordered as $\lambda_0 > \lambda_1 > \cdots > \lambda_{n-1}$.
\end{definition}

The (expected) stochastic discounted payoff for principal~$i$ under strategy~$\sigma$, starting from initial state~$s$, over a run $\langle s_0 = s, a_0, s_1, a_1, s_2, a_2, \ldots \rangle$, is defined as:
\begin{equation}
\EDisct^{\sigma}_{i}(s) = \eE \left[ \sum_{j=0}^\infty \lambda_i^j \cdot R(s_j, a_j, i) \right].
\label{player_i_payoff}
\end{equation}

\begin{assumption}\label{assum-spaced}
Throughout this work, we assume that the discount factors are \emph{reasonably spaced}, meaning that for each pair $\lambda_i, \lambda_{i+1} \in \Lambda$, the quantity $\frac{1}{\lambda_i / \lambda_{i+1} - 1}$ is bounded by a polynomial in the size of the MDP. This assumption is strictly weaker than requiring unary encoding of discount factors or bounding their denominators by a polynomial function of the input size. 
In practice, discount factors often arise from empirical studies and are typically rounded to a fixed precision or selected from a small, predefined set: for instance, with denominators capped at 1,000. All such cases naturally yield reasonably spaced discount factors. With this assumption, we exclude cases in which two discount factors are too close. For more clarification, an example is provided in~\Cref{sec:example-spaced}.

\end{assumption}

\subsection{Social Welfare}
When an AI agent acts on behalf of multiple principals, each with distinct time preferences and utility functions, evaluating the overall benefit of a shared strategy requires an aggregate measure. The \emph{social welfare} of a strategy captures the total discounted payoff across all principals, providing a utilitarian benchmark for collective optimality.

Let $n$ be the number of principals, $s$ the initial state, and $\sigma \in \Sigma$ the strategy followed by the MDP~$\mdp$. The social welfare $\sw_{\sigma}(s)$ under strategy~$\sigma$ from state $s$ is defined as:
\begin{equation}
\sw_{\sigma}(s) = \sum_{i=0}^{n-1} \EDisct^{\sigma}_{i}(s).
\label{SoWe}
\end{equation}

Unfolding the definition of expected discounted payoff from Equation~\eqref{player_i_payoff}, we obtain:
\begin{equation}
\sw_{\sigma}(s) = \sum_{i=0}^{n-1} \mathbb{E} \left[ \sum_{j=0}^\infty \lambda_i^j \cdot R(s_j, a_j, i) \right].
\label{open-SoWe}
\end{equation}
The optimal social welfare $\sw_*$ from $s \in S$ is defined as: 
\[
\sw_*(s) = \sup_{\sigma \in \Sigma} \sw_{\sigma}(s).
\]
We say that a strategy $\sigma^*$ is welfare-optimal if 
\[
\sw_{\sigma_*}(s) = \sw_*(s).
\]

\begin{remark}
 We note that social welfare is a robust objective since it can encode any objective that is a linear combination of the individual expected payoffs by just factoring in the linear factor into their respective rewards. This includes natural cases such as when one player's objective has higher priority than the other expressed appropriately in the linear function; or, in the case where several players have same discount factor, all merged into one player.   
\end{remark}

In the remainder of the paper, we study how to compute strategies that maximise social welfare in asymmetrically-discounted MDPs. Section~\ref{sec:pos-hard} analyses the power and limitations of stationary and positional strategies, and establishes the computational hardness of identifying socially optimal strategies within these classes. Section~\ref{sec:algo} then presents an efficient algorithm for synthesising optimal strategies that maximise the utilitarian objective. Finally, Section~\ref{sec:experiment} evaluates the scalability of our algorithm with respect to the number of states, the number of principals, and the ratio between discount factors.

\section{Stationary Strategy Synthesis is Hard}
\label{sec:pos-hard}
As a first step, we investigate computing welfare optimal strategies with simpler strategies such as stationary or positional strategies.
We show that the problem of deciding whether there is a stationary strategy $\sigma$ such that $ \sw_\sigma \geq \kappa$ for a given threshold $\kappa$, is NP-hard.

\begin{theorem}\label{pos-NP}
The problem of determining whether in an asymmetrically-discounted MDP, the optimal social welfare is greater or equal to a given threshold is NP-hard for stationary strategies, even when restricted to positional strategies. The problem is NP-hard even when there are only two principals and all the rewards are same for all principals, or all the rewards are zero-sum for the two principals.
\end{theorem}

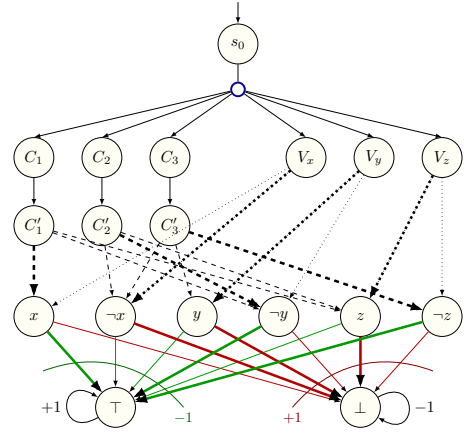
\begin{figure}[t!]
\centering
\scalebox{0.6}{
\begin{tikzpicture}[
  every state/.append style={fill=safecellcolor},
  every node/.append style={box state/.append style={fill=safecellcolor}},
  initial where=above,
  initial text={},
  true edge/.style={->, draw=green!60!black},
  false edge/.style={->, draw=red!70!black},
  clause edge/.style={->, dashed},
  var edge/.style={->, dotted}
]

  \node[ state, initial] (q_0) {$s_0$};

  \node[prob state] (prob) [below=1cm of q_0] {};

  \node[ state] (c_1) [below left=1.5cm and 4.5cm of prob] {$C_1$};
  \node[ state] (c_2) [below left=1.5cm and 3cm of prob]     {$C_2$};
  \node[ state] (c_3) [below left=1.5cm and 1.5cm of prob] {$C_3$};

\node[ state] (c_1_prime) [below=1.5cm of c_1] {$C_1^\prime$};
\node[ state] (c_2_prime) [below=1.5cm of c_2] {$C_2^\prime$};
\node[ state] (c_3_prime) [below=1.5cm of c_3] {$C_3^\prime$};

  \node[ state] (z) [below right=1.5cm and 4.5cm of prob]  {$V_z$};
  \node[ state] (y) [below right=1.5cm and 3cm of prob]    {$V_y$};
  \node[ state] (x) [below right=1.5cm and 1.5cm of prob]  {$V_x$};

  \foreach \i/\name/\labeltext in {
    1/x_var/x,
    2/not_x_var/\lnot x,
    3/y_var/y,
    4/not_y_var/\lnot y,
    5/z_var/z,
    6/not_z_var/\lnot z
  } {
    \node[ state] (\name) at ({(\i - 3.5) * 1.8}, -6) {$\labeltext$};
  }

  \node[state] (true) at (-2.7, -8) {$\top$};
  \node[state] (false) at (2.7, -8) {$\bot$};

  \path[-] (q_0) edge node {} (prob);

  \path[->] (prob) edge (c_1.north)
                   edge (c_2.north)
                   edge (c_3.north);

  \path[->] (prob) edge (x.north)
                   edge (y.north)
                   edge (z.north);

  \path[->] (c_1) edge (c_1_prime)
            (c_2) edge (c_2_prime)
            (c_3) edge (c_3_prime);

  \path[clause edge] 
    (c_1_prime) edge (x_var)
          edge (not_y_var)
          edge (z_var);

  \path[clause edge] 
    (c_2_prime) edge (not_x_var)
          edge (not_y_var)
          edge (z_var);

  \path[clause edge] 
    (c_3_prime) edge (not_x_var)
          edge (y_var)
          edge (not_z_var);

  \path[var edge]
    (x) edge (x_var)
         edge (not_x_var);

  \path[var edge]
    (y) edge (y_var)
         edge (not_y_var);

  \path[var edge]
    (z) edge (z_var)
         edge (not_z_var);

  \foreach \name in {x_var, not_x_var, y_var, not_y_var, z_var, not_z_var} {
    \path[true edge] (\name) edge node[left] {} (true);
    \path[false edge] (\name) edge node[right] {} (false);
}

  \path[->,loop left] (true) edge node {$+1$}  (true);
  \path[->,loop right] (false) edge node {$-1$} (false);

  \coordinate (right)  at (-1.2, -8);
  \coordinate (left) at (-4.3, -7.2);
  \draw[green!40!black] (right) to[bend left=-35] node[pos=0, below]{\footnotesize $-1$} (left);

  \coordinate (right)  at (1.2, -8);
  \coordinate (left) at (4.3, -7.2);
  \draw[red!50!black] (right) to[bend left=35] node[pos=0, below]{\footnotesize $+1$} (left);

  \path[clause edge, ultra thick] 
    (c_1_prime) edge (x_var);

  \path[clause edge, ultra thick] 
    (c_2_prime) edge (not_y_var);

  \path[clause edge, ultra thick] 
    (c_3_prime) edge (not_z_var);

  \path[var edge, ultra thick]
    (x) edge (not_x_var);

  \path[var edge, ultra thick]
    (y) edge (y_var);

  \path[var edge, ultra thick]
    (z) edge (z_var);

    \foreach \name in {not_x_var, y_var, z_var} {
    \path[false edge, ultra thick] (\name) edge node[right] {} (false);
  }

  \foreach \name in {x_var, not_y_var, not_z_var} {
    \path[true edge, ultra thick] (\name) edge node[left] {} (true);
  }
  
\end{tikzpicture}
}
\caption{
The MDP constructed from 3-SAT formula $\phi = (x \lor \neg y \lor z) \land (\neg x \lor \neg y \lor z) \land (\neg x \lor y \lor \neg z)$. 
 Curved arcs indicate rewards associated with groups of edges of same color.    
Edges without labels carry reward~$0$. Probabilistic nodes have outgoing transitions with uniform probability.
 The thick edges from $C'_i$ depicts a satisfying assignment.}
\label{fig:pos-NP}
\vspace{-1.2em}
\end{figure}

\begin{proof}[Proof sketch]
We provide a reduction from 3-SAT (\cite{Karp1972}). Given an instance of 3-SAT, a Boolean formula $\phi$, our reduction constructs an asymmetrically-discounted MDP $\mdp_\phi$ with two principals, such that the optimal social welfare is non-negative in $\mdp_\phi$ if and only if $\phi$ is satisfiable. See \Cref{fig:pos-NP} for an illustrative example of this reduction. Given a 3-SAT instance $\phi$ over a set $X = \set{x_1, x_2, \ldots, x_n}$ of $n$ variables and with $m$ clauses $C = \set{C_1, C_2, \ldots, C_m}$, we describe the MDP $\mdp_\phi$. Let $L = X \cup \{\neg x | x \in X\}$ be the set of all literals over $X$.
$\mdp_\phi$ has an initial state $s_0$ from which there is a single action, which with equal probability $\frac{1}{n+m}$ leads to a state from the set $C \cup \set{V_x: x \in X}$. This is the only action in $\mdp_\phi$ with probabilistic transition. From each state $C_i \in C$, there is a single action which leads to its duplicate state $C_i^{\prime}$ in $C^{\prime} = \set{C_1^{\prime}, C_2^{\prime}, \ldots, C_m^{\prime}}$. 
For each state $C_i^{\prime}$, there is an action 
for each literal $\ell$ in clause $C_i$ and leads to a state $\ell$ from the set of states $L$. 
Similarly for each state $V_x$ for variable $x$, there are actions for both literal $x$ and $\neg x$, leading to respective states $x$ and $\neg x$. The game has two sink states $\top$ and $\bot$ which has single trivial actions that stays on respective states. From each state $l \in L$, there are two actions that chooses between $\top$ and $\bot$. $\mdp_\phi$ has two principals with discount factors $0.54$ and $0.4$ respectively. The reward to both the principals are same. At each state $l \in L$, they receive $-1$ for choosing $\top$ and $+1$ for choosing $\bot$; at states $\top$ and $\bot$ they receive $+1$ and $-1$ respectively for staying in; they receive $0$ for every other action.    

The construction ensures that satisfying assignments correspond to strategies that direct short paths to low-reward sinks and long paths to high-reward sinks, aligning reward trajectories with the structure of~$\phi$. Carefully chosen discount factors ensure that only such satisfying strategies yield a social welfare above a given threshold.
\end{proof}
The complete proof details can be found in Appendix~\ref{sec:proof-full}.

\section{Computing Welfare-Optimal strategies}
\label{sec:algo}

In this section, we present a polynomial-time algorithm for synthesizing a joint strategy that maximises social welfare in a given asymmetrically-discounted MDP.

We begin by outlining our algorithm and the underlying intuition. The key insight is that there always exists a positional strategy that can be adopted after a finite number of steps without loss of optimality. Our algorithm first identifies such a \emph{limit strategy}.
Next, we show that the number of steps required before switching to this limit strategy is polynomial in the size of the MDP. By unrolling the MDP into a directed acyclic graph (DAG) up to this horizon and evaluating it backward, we obtain a simple and efficient construction of welfare-optimal strategies. These strategies are pure, finite-memory, and rely only on a counter.
We conclude the section by establishing that computing the optimal social welfare in asymmetrically-discounted MDPs is P-complete.

 \begin{algorithm}[t!]{}  
 \caption{Compute long-term strategy for optimal welfare}
 \label{algo:long-term}
 \begin{algorithmic}[1]
\STATE \textbf{Input} : $\mdp$
\STATE \textbf{Output}: MDP with long-term strategies and evaluation
\STATE $\mdp_0 \gets \mdp$
\FOR {$j \in \{0,\cdots,n-1\}$}
\STATE $V_j \gets$ optimal values for Principal $j$ in $\mdp_j$ 
\STATE $\mdp_{j+1} \gets$ $\mdp_j$ restricted to actions with a strategy optimal for Principal $j$ in $\mdp_j$
\ENDFOR
\RETURN  $\mdp_n,V_0,V_1,\ldots,V_{n-1}$
 \end{algorithmic} 
 \end{algorithm}

\subsection{Asymptotic Strategy Behaviour}
We begin by computing a long-term strategy by solving the MDP from the perspective of Principal~$0$—the most patient principal (highest discount factor). 
Since this reduces to a standard single-agent MDP with a discounted reward objective, optimal state values and corresponding actions can be computed in polynomial time~\cite{Put94}.
In general, multiple actions may yield the same optimal value in a given state. To resolve such ties, we sequentially invoke the preferences of the remaining principals as tie-breakers: first Principal~$1$, then Principal~$2$, and so on, following the order of increasing patience. This is formalized in Algorithm~\ref{algo:long-term}.

Let $V_i \colon S \to \mathbb{Q}$ denote the value function for Principal~$i$ in the MDP $\mdp_i$, where $\mdp_i$ is the MDP restricted to the remaining admissible actions after resolving ties at the previous levels.
The reason for correctness is that, due to faster discounting, an advantage of a principal with a lower discount factor will eventually be dwarfed by a disadvantage of a principal with a higher discount factor.
We do not have to prove this at this point; instead, we will later show that deviation from this strategy will eventually be unattractive.

\begin{lemma}\label{lem:longTermPoly}
Algorithm \ref{algo:long-term} runs in time polynomial in its input.
\end{lemma}

This follows from the fact that the procedure evaluates n MDPs, each of which can be evaluated in polynomial time~\cite{Put94}.

\subsection{Deviations from Asymptotic Behaviour}
As illustrated by \cref{ex:needmem} in the introduction—where the long-term strategy prescribes always selecting action $b$—it can be advantageous for social welfare to initially deviate from the asymptotic strategy. Specifically, a deviation in the very first step (i.e., after zero steps) may yield improved outcomes for certain principals.
To quantify the impact of such a deviation, we define the payoff difference for Principal~$i$ resulting from taking action~$a$ in state~$s$, instead of following the long-term strategy. This difference~\cite{baird1993advantage,Sutton18} is given by:
\begin{equation}
\Delta_0(s,a,i) \mapsto R(s, a, i) + \lambda_i \big(\sum_{s' \in S} p(s' \mid s, a) \cdot V_i(s')\big) - V_i(s)
\label{raw-state-action-delta}
\end{equation}
where $V_i(s)$ denotes the value of state $s$ under the optimal strategy for Principal~$i$ in the corresponding one-principal MDP~$\mdp_i$.
If the same happens after $j$-steps, the payoff difference is
\begin{equation}
\Delta_j(s,a,i) \mapsto {\lambda_i}^j \cdot \Delta_0(s, a, i)
\label{raw-state-action-delta}
\end{equation}
Note that we now have
\begin{equation}\label{eq:social-sums}
\sw_{\sigma}(s) =\sum_{i=0}^{n-1}V_i(s) + \sum_{j=0}^\infty \mathbb E \sum_{i=0}^{n-1} \Delta_j(s_j,a_j,i)
\end{equation}

We now only need to argue that, for reasonably spaced discount factors, there is a small $\kappa$ s.t., for all $j>\kappa$, $\sum_{i=0}^{n-1} \Delta_j(s_j,a_j,i) \leq 0$ holds.
Once this is the case, we can optimise social welfare by never deviating from $\mdp_n$ from this point onwards.

\begin{observation}
For all states $s\in S$ and all actions $a \in A(s)$ for $\mdp_n$, all $i\in P$, and $j\in \mathbb N_0$, we have $\Delta_j(s,a,i)=0$.
For all actions $a \in A(s)$ for $\mdp$, but not for $\mdp_n$, we have that there is a minimal $i$ such that $\Delta_0(s,a,i) \neq 0$; and for this $i$, $\Delta_0(s,a,i) < 0$ holds.
\end{observation}

\begin{lemma}\label{lem:initial-sums}
For all $s\in S, a\in A(s), i \in P, j\in \mathbb N_0$, we have
\begin{multline*}
\forall i' \leq i, 
\sum_{p=0}^{i'}\Delta_j(s,a,p)\leq 0 \implies \\
\forall j' \geq j, \text{ and } \forall i' \leq i \text{ we have } \sum_{p=0}^{i'}\Delta_{j'}(s,a,p)\leq 0.    
\end{multline*}
\end{lemma}
\begin{proof}
This can be shown by a simple inductive argument with trivial basis ($k = j$) and induction step $k \mapsto k+1$ as
\[\begin{array}{cl}
\hspace*{-2mm}\sum_{p=0}^{\ell}\Delta_{k+1}(s,a,p)  = & \hspace*{-3mm}\sum_{\ell' = 0}^{\ell-1} (\lambda_{\ell'}{-}\lambda_{\ell'+1}) \sum_{p=0}^{\ell'}\Delta_{k}(s,a,p) \\
& + \lambda_\ell\sum_{p=0}^{\ell}\Delta_{k}(s,a,p) \leq 0 
\end{array}\]
for $\ell \leq i$, where the inner sums are $\leq 0$ by induction hypothesis.
\end{proof}

As the $V_i$ returned by Algorithm \ref{algo:long-term} are the solution to a system of linear equations defined by any remaining strategy in $\mdp_n$, their values are fractions with both the nominator and the denominator singly exponential (in value, and polynomial in length) in the input to the algorithm.
Thus, to estimate $\kappa$, when we take such a $\Delta_0(s,a,i) < 0$, then $\kappa_{s,a}'=\sum_{j=i+1}^{n-1} \max\{0,\Delta_0(s,a,j)\}/|\Delta_0(s,a,i)|$ , is exponentially bounded by value, and polynomially by length, in the input to Algorithm \ref{algo:long-term}.
It is now easy to see that%
\footnote{This skips over the corner case of $i=n-1$; in that case, $\sum_{i=0}^{n-1} \Delta_j(s,a,i) < 0$ holds for all $j\in \mathbb N_0$ and we set $\kappa_{s,a}=0$.},
for
\begin{equation}\label{eq:kappa-calculate}
\kappa_{s,a} 
=\big\lceil\log_{\lambda_i/\lambda_{i+1}} \kappa_{s,a}'\big\rceil
=\big\lceil\frac{ \ln \kappa_{s,a}'}{\ln(\lambda_i/\lambda_{i+1})}\big\rceil
\end{equation}
and $j \geq \kappa_{s,a}$, $\sum_{i=0}^{n-1} \Delta_j(s,a,i) \leq 0$ holds.
For an action $a \in A(s)$ for $\mdp_n$, we have $\Delta_j(s,a,i) = 0$ for all principals $i$ and all $j\in \mathbb N_0$, so we choose $\kappa_{s,a}=0$. We set $\kappa  = \max\{\kappa_{s,a}\mid s \in S, a \in A(s)\}$ to the maximal value among the $\kappa_{s,a}$, so that, for all $j \geq \kappa$, $\Delta_j(s,a,i)\leq 0$ holds.

\begin{lemma}\label{lem:notTooLong}
A social welfare cannot be obtained when deviating from the strategies available in $\mdp_n$ after $\kappa$ steps (as defined above).
If the discount factors are reasonably spaced, $\kappa$ is polynomial in the size of the input to Algorithm~\ref{algo:long-term}.
\end{lemma}

\begin{proof}
For the first part of the claim, we note that $\kappa$ is chosen so that the initial sums from Lemma \ref{lem:initial-sums} are all non-positive.
Thus, for $j \geq \kappa$ steps, all expected sums $\mathbb E \sum_{i=0}^{n-1} \Delta_j(s_j,a_j,i)$ from Equation \ref{eq:social-sums} are non-positive and are zero if, and only if, the chance of taking an action available in $\mdp_n$ is one.
For the second part, we have argued that each $\ln \kappa_{s,a}'$ is polynomially bounded (\Cref{assum-spaced}). For the denominator, we have that $\ln ((\lambda_i/\lambda_{i+1})) \approx \frac{1}{(\lambda_i/\lambda_{i+1})-1}$ when $(\lambda_i/\lambda_{i+1})$ is close to $1$ as $(1+ \frac{1}{k})^k \approx e$.
\end{proof}

\subsection{Computing Social Welfare}
In order to find the social welfare (and a strategy that achieves it), we consider Algorithm \ref{algo:long-term} again and note that we have already determined $\sum_{i=0}^{n-1}V_i(s)$ and showed that we have to follow $\mdp_n$ after $\kappa$ steps as defined in the previous section---where we also showed that $\sum_{j=\kappa}^\infty \mathbb E \sum_{i=0}^{n-1} \Delta_j(s_j,a_j,i) = 0$ for social welfare.
To maximise social welfare, maximise:
\begin{equation}\label{eq:initial-sums}
\sum_{j=0}^{\kappa-1} \mathbb E \sum_{i=0}^{n-1} \Delta_j(s_j,a_j,i) \ .
\end{equation}
To do this, we can simply unravel the MDP $\mdp = (S, A, T, P, R, \Lambda
)$ to a finite DAG 
\[
\mathcal D = (S \times \{0,\ldots,\kappa{-}1\}, A', T', P, R', \Lambda
)\]
on the fly, where, for $j<\kappa -1$,
$A'(s,j)= A(s)$,
$T'(s,j,a)(s',j+1) = T(s,a)(s')$, and
$R'(s,j,a)= \sum_{i\in P}\Delta_j(s,a,i)$, while
$A'(s,\kappa-1)=\emptyset$.
Note that we do not calculate $\kappa$, but stop unravelling once, for all $s \in S$ and all $a \in A(s)$ all initial sums of $\sum_{i =0}^{i'}\Delta_j(s,a,i)$ are non-positive, which Lemma \ref{lem:initial-sums} shows to henceforth holds forever, while Lemma \ref{lem:notTooLong} shows that this happens after few iterations.

\begin{lemma}\label{lem:correct}
An optimal total reward strategy $\sigma$ for the DAG MDP $\mathcal{D}$ (followed by any stay-in-$\mdp_n$ strategy) defines a finite counting strategy $\sigma'$ for $\mdp$ that provides optimal social welfare.
Moreover, if its expected reward for the initial strategy on $\mathcal D$ for a state $(s,0)$ is $E_\sigma(s,0)$, then the
the social welfare for a starting state $s$ is $\sw_{\sigma'} = \sum_{i=0}^{n-1}V_i(s) + E_\sigma(s,0)$.
\end{lemma}

\begin{proof}
By construction of $\mathcal D$, its optimal solutions maximise $\sum_{j=0}^{\kappa-1} \mathbb E \sum_{i=0}^{n-1} \Delta_j(s_j,a_j,i)$, and $\kappa$ is selected so that, by Lemma \ref{lem:initial-sums}, $\sum_{i=0}^{n-1} \Delta_j(s,a,i)\leq 0$ for all $j\geq \kappa$, while following $\mdp_n$ provides $\sum_{i=0}^{n-1} \Delta_j(s,a,i) = 0$.
The value for $\sw_{\sigma'}$ can be obtained from Equation~\ref{eq:social-sums}.
\end{proof}

\subsection{Complexity}
The preceding results yield the following complexity classification. 
For reasonably spaced discount factors, welfare optimisation is polynomial-time solvable, while the associated decision problem remains \mbox{P-complete}. 
Here, \textsc{SocialWelfare} asks whether there exists a strategy~$\sigma$ whose social welfare is non-negative.

\begin{theorem}\label{thm:main}
The social welfare of an asymmetrically-discounted MDP $\mdp$ with reasonably spaced discount factors can be computed in polynomial time.
Moreover, \textsc{SocialWelfare} is \mbox{P-complete}, even when restricted to two-player zero-sum games or single-player MDPs.
\end{theorem}

\begin{proof}
For inclusion in $P$, we have shown that our algorithm terminates with the social welfare (cf.\ Lemma \ref{lem:correct}.
To obtain this, we first run Algorithm \ref{algo:long-term} to obtain the long-term values and strategy time polynomial in the input to Algorithm \ref{algo:long-term} (cf.\ Lemma \ref{lem:longTermPoly}).
We then unravel $\mdp$, replacing the reward after $j$ steps by $R(s,a,j) = \sum_{i=0}^{n-1} \Delta_j(s,a,i)$.
Each unravelling step is cheap, and for reasonably spaced discount factors, we are guaranteed to stop unravelling after a polynomial number of steps (cf.\ Lemma \ref{lem:notTooLong}) into an MDP $\mathcal D$ with a DAG structure.

Analysing $\mathcal D$ can then be done layer by layer: for states in the final layer $\kappa$, the expected payoff is $E(s,\kappa)= 0$ (as the end of the DAG has been reached).
Once a layer $j+1$ is evaluated, we can locally evaluate 
\[
E(s,j) {=} \max_{a \in A(s)} \{R(s,a,j) {+} \sum_{s'\in S}T(s,j,a)(s',j{+}1)E(s',j)\}.
\]
This can be done in time polynomial in $\mathcal D$ (and in $\mdp$) when discount factors are reasonably spaced (cf.\ Lemma~\ref{lem:notTooLong}).

For hardness, we reduce from reachability games~\cite{DBLP:books/daglib/0095988} 
on alternating graphs.
In an alternating graph, the reachability player controls the $\OOR$ nodes, while the safety player controls the $\AAND$ nodes.
We translate the game nodes to states, translating $\OOR$ nodes to states with one actions, so that $T$ picks among the successors of that node uniformly at random.
We translate $\AAND$ nodes to states with as many actions as it has successors and let the MDP freely pick to which successor state it moves.
For payoffs, we use $R(t,a) = -1$ for all $a \in A(t)$, and $R(s',a) = 0$ otherwise.

If the safety player can win, she can win with a positional strategy in the reachability game; the same strategy provides a social welfare of $0$, which is optimal as no individual payoff is positive.
Vice versa, a strategy $\sigma$ that provides $\sw_\sigma = 0$ in the resulting MDP provides a winning strategy for the safety player in the reachability game.
Finally, if we consider two-player zero-sum games instead, we use the payoffs $R(t,a,0) = -1$, $R(t,a) = -$, and $R(t,a,i) = 0$ otherwise and retain the same argument.
\end{proof}

\section{Experimental Results}
\label{sec:experiment}
We study scalability through three research questions (\textbf{RQ}s): 
\textbf{(RQ1)} the number of states,
\textbf{(RQ2)} the number of principals, and
\textbf{(RQ3)} the spacing of discount factors.
For each question, we construct randomised MDP families 
with structural features designed to expose the relevant performance trends.
The results demonstrate the empirical tractability of our method across these key dimensions of complexity.

All experiments were run on a Linux workstation with an
Intel\textsuperscript{\textregistered} Core\textsuperscript{\texttrademark}
i7-4790 CPU at 3.60\,GHz, 8 logical cores, and 8.2\,GB RAM,
using Python~$\mathtt{3.12.3}$, Linux kernel~$\mathtt{6.14.0{-}24{-}generic}$,
and $\mathtt{glibc~2.39}$.

\begin{figure}[t]
\centering
\caption{Running time grows linearly with states in reasonably sized MDPs with up to $2K$ states (top-left) and (without \Cref{algo:long-term}) in MDPs with a wide variation ($3$ to $10^5$) of states (top-right); linearly with the number of players (bottom-left); and falls as the factor between discount factors grows (bottom-right).}
\label{fig:combined_time}
\vspace{-1em}
\begin{subfigure}{0.48\linewidth}
    \centering
    \includegraphics[width=\linewidth]{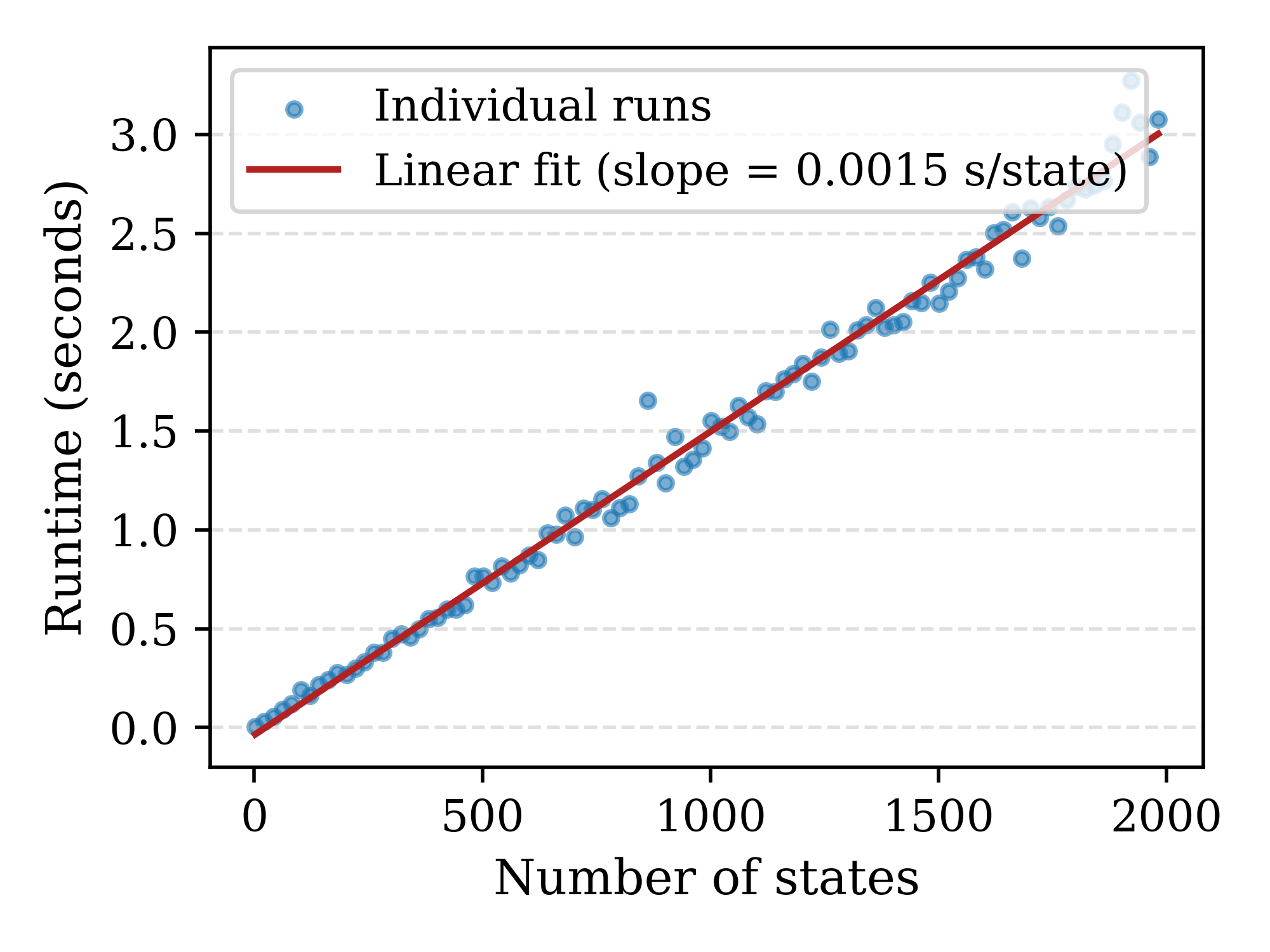}
\end{subfigure}\hfill
\begin{subfigure}{0.48\linewidth}
    \centering
    \includegraphics[width=\linewidth]{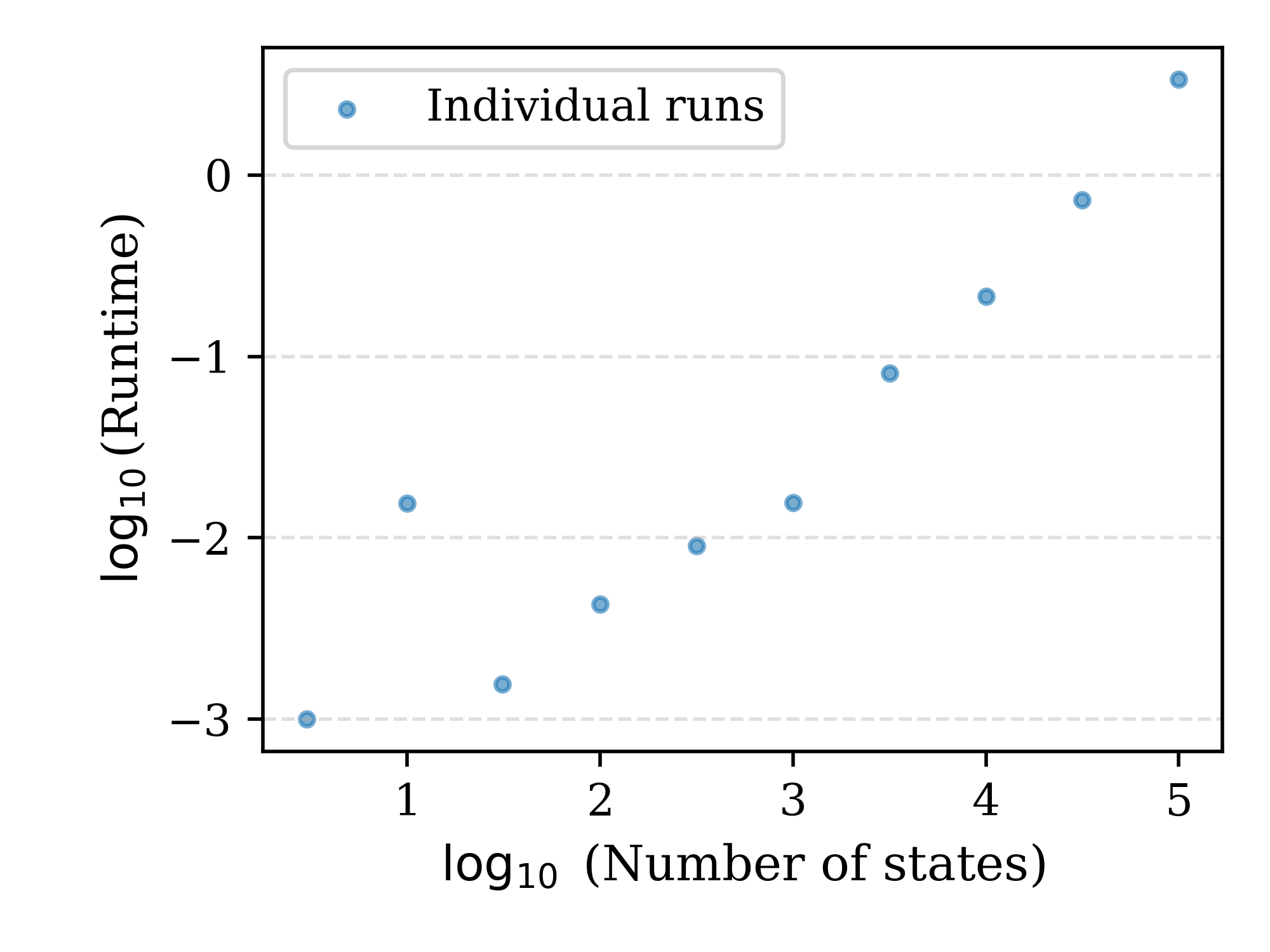}
\end{subfigure}

\begin{subfigure}{0.48\linewidth}
    \centering
    \includegraphics[width=\linewidth]{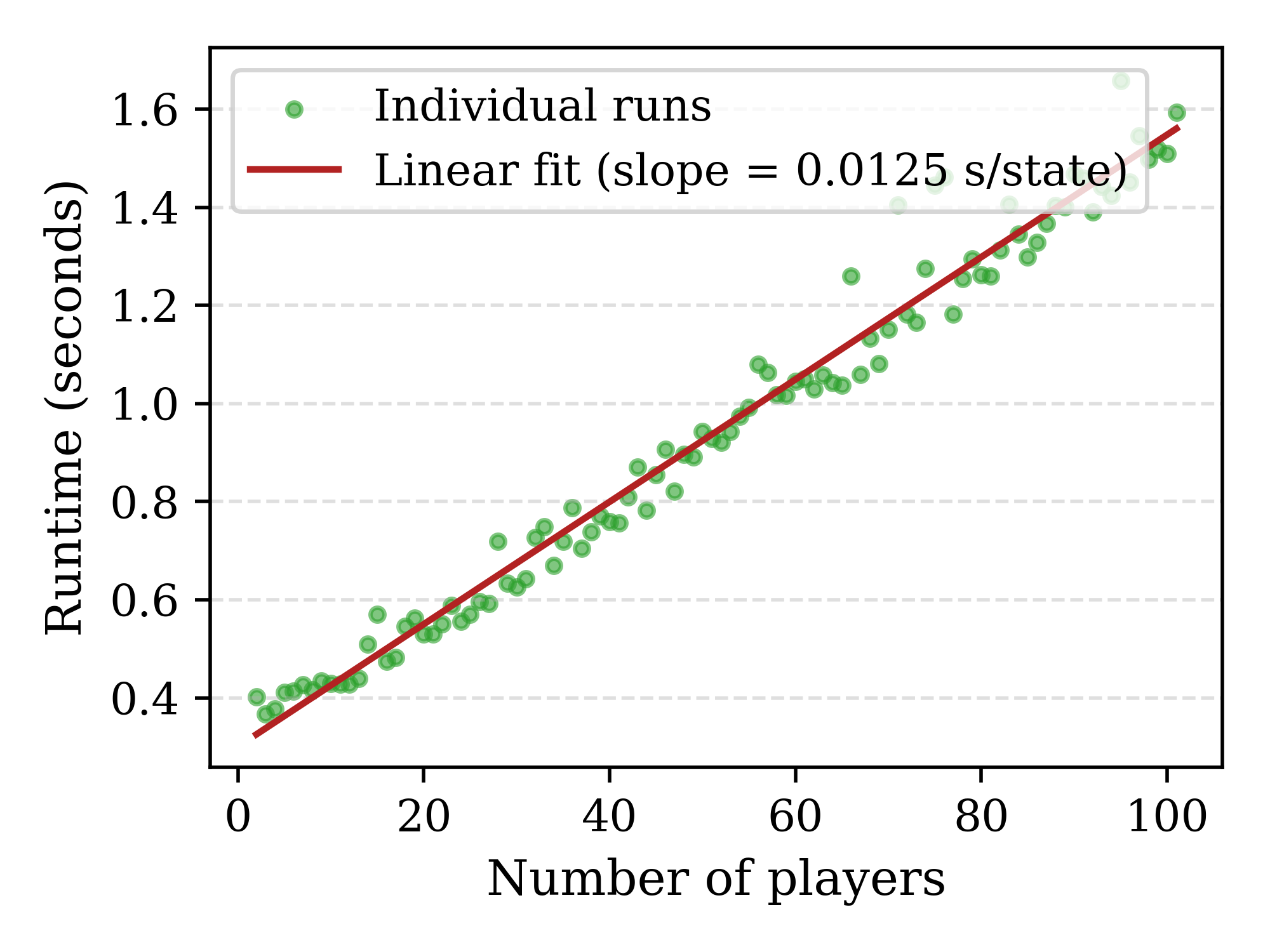}
\end{subfigure}\hfill
\begin{subfigure}{0.48\linewidth}
    \centering
    \includegraphics[width=\linewidth]{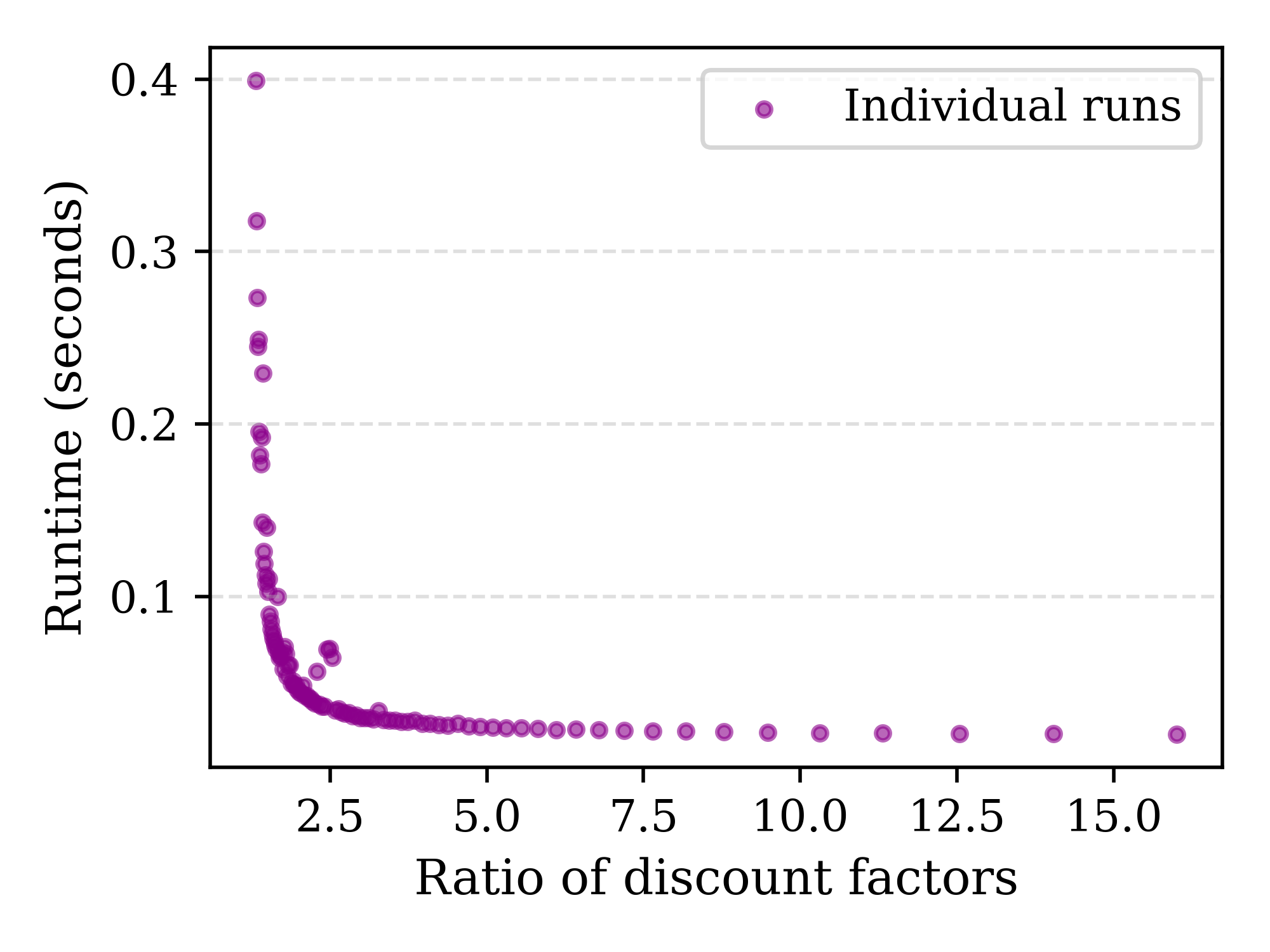}
\end{subfigure}

\end{figure}

\paragraph{RQ1: Scalability with size of MDPs.}
To evaluate scalability with respect to the number of states, we designed two sets of experiments: first, we generated 100 randomized MDP instances, varying the number of states from $2$ to $2000$ in increments of $20$. Each MDP has two principals with fixed discount factors $(0.9, 0.3)$, and two actions per state.
\Cref{fig:combined_time} (top-left) shows that as the number of states in the MDPs increases (from $2$ to $2000$), the time required to calculate the optimal social welfare grows linearly with the size of the MDP, demonstrating scalability in relation to the number of states. 
Second, to assess scalability on large-scale MDPs, we evaluated the algorithm on a smaller set of instances with larger state spaces. We generated MDPs with the number of states ranging from $3$ to $10^5$, using powers of $\sqrt{10}$. Each MDP consisted of 6 principals with fixed discount factors and 2 actions per state. \Cref{fig:combined_time} (top-right) shows the algorithm remains tractable even for large MDPs; only in this figure, we excluded the runtime of \Cref{algo:long-term} because solving the MDPs dominates the cost and becomes bottleneck, see \Cref{experiments}.

\paragraph{RQ2: Scalability with the number of principals.}
To evaluate scalability with respect to the number of principals, we generated $100$ randomized MDP instances, varying the number of principals from $2$ to $101$ in increments of 1 principal per example. Each MDP had fixed number of $30$ states, and exactly two actions per state. The principals’ discount factors formed an arithmetic progression from $0.99$ (most patient) to $0.05$ (least patient) according to \(\lambda_i = 0.99 - (i-1)\tfrac{0.94}{n-1}\). \Cref{fig:combined_time} (bottom-left) shows that the computational time required to compute optimal social welfare grows polynomially with the number of principals of MDP, demonstrating scalability of the algorithm in relation to the number of principals.

\paragraph{RQ3: Scalability with the ratio of discount factors.}
To evaluate scalability with respect to the ratio of discount factors of principals ($\,\lambda_0\,/\,\lambda_1$), we generated $100$ randomized MDP instances, varying the ratio of discount factors from $1.32$ (almost identical patience) up to $16$ (strongly asymmetric patience). Each MDP had two principals, fixed number of 30 states, and exactly two actions per state.
\Cref{fig:combined_time} (bottom-right) shows that the execution time falls steeply as the ratio grows: runs start at $\approx 0.3$s when the two discount factors are nearly equal ($\frac{\lambda_0}{\lambda_1} =1.32$), drop below $0.02$s when $\frac{\lambda_0}{\lambda_1} \approx 4$, and then flatten out around $0.01$s for when $\frac{\lambda_0}{\lambda_1} = 16$.

\section{Conclusion}
\label{sec:concl}
We introduced asymmetrically-discounted MDPs as a natural model for decision-making on behalf of multiple principals with heterogeneous time preferences. We showed that social-welfare-optimal strategies can be computed tractably under reasonable assumptions on discount-factor spacing, and that such optima can be realised by pure, finite-memory strategies.

Our results highlight the computational and strategic complexity introduced by temporal asymmetries, and open new directions for multi-agent planning, algorithmic game theory, and the study of intertemporal incentives.

\section*{Acknowledgements}
This work was supported by the EPSRC through grants EP/X03688X/1 (TRUSTED) and EP/X042596/1 (Games for Good), and in part by the NSF under CAREER Award CCF-2146563. Ashutosh Trivedi is a Royal Society Wolfson Visiting Fellow and gratefully acknowledges the support of the Wolfson Foundation and the Royal Society.

\bibliographystyle{named}
\bibliography{ijcai26}

\newpage
\newpage
\appendix
\onecolumn

\setcounter{secnumdepth}{2} 

\begin{center}
\textbf{\LARGE Appendix}    
\end{center}

The appendix collects supplementary material supporting the main technical and experimental claims of the paper. 
We begin with a visual summary of the main results, then give additional details for the theoretical analysis, including an example illustrating the role of reasonably spaced discount factors and the omitted proof from Section~\ref{sec:pos-hard}. 
We next provide enlarged experimental plots and additional runtime comparisons. 
Finally, we include worked examples that illustrate how the algorithm computes welfare-optimal strategies in concrete asymmetrically-discounted MDPs.

\section{Summary of Results}
\label{app:summary}

\begin{figure}[!htbp]
\scalebox{1}{
\begin{tikzpicture}[scale=1]

\filldraw[
  fill=blue!20,
  draw=black,
  very thick
]
(0,0) ellipse (5 and 3);

\filldraw[
  fill=green!25,
  draw=black,
  thick
]
(0,0) ellipse (3.5 and 2);

\filldraw[
  fill=orange!35,
  draw=black,
  thick
]
(0,0) ellipse (1.8 and 1);

\begin{scope}
  \clip (0,0) ellipse (5 and 3);
  \fill[
    pattern=north east lines,
    pattern color=black
  ] (-6,-4) rectangle (0,4);
\end{scope}

\draw[very thick] (0,-3) -- (0,3);

\node[
  fill=white,
  inner sep=2pt
] (pure) at (1,2.5) {Pure};

\node[
  fill=white,
  inner sep=2pt
] (mixed) at (-1,2.5) {Mixed};

\node[
  fill=white,
  inner sep=2pt
] (counting) at (0,1.5) {Counting};

\node[
  fill=white,
  inner sep=2pt,
  align=center
] (pureCounting) at (2.6,0) {Pure\\Counting};

\node[
  fill=white,
  inner sep=2pt
] (stationary) at (0,.6) {Stationary};

\node[
  fill=white,
  inner sep=2pt
] (positional) at (0.9,0) {Positional};

\node[left,align=center] (optimality) at (-5.5,2) {\textbf{WELFARE}\\\textbf{OPTIMALITY}};

\node[right] (complexity) at (5.5,2) {\textbf{COMPLEXITY}};

\node[left] (no1) at (-5.5,1.2) {No};

\node[left,align=center] (no2) at (-5.5,0.4) {No~\\(Even within\\stationary strategies)};

\node[left] (yes) at (-5.5,-.6) {Yes};

\node[right] (np1) at (5.5,1.2) {NP-hard};

\node[right] (np2) at (5.5,0.4) {NP-hard};

\node[right] (pcomplete) at (5.5,-.4) {P-Complete};

\draw[
  ->,
  thick
]
(stationary.west)
  .. controls +(-1.8,1.2) and +(1.2,0.3) ..
(no1.east);

\draw[
  ->,
  thick
]
(positional.west)
  .. controls +(-1,0) and +(1.2,0.3) ..
(no2.east);

\draw[
  ->,
  thick
]
(pureCounting.west)
  .. controls +(-1.8,-1.8) and +(1.2,-0.3) ..
(yes.east);

\draw[
  ->,
  thick
]
(stationary.east)
  .. controls +(1.8,1.2) and +(-0.7,0.1) ..
(np1.west);

\draw[
  ->,
  thick
]
(positional.north)
  .. controls +(1,1) and +(-0.7,0.5) ..
(np2.west);

\draw[
  ->,
  thick
]
(pureCounting.east)
  .. controls +(.8,-.8) and +(-0.2,-0.3) ..
(pcomplete.west);

\end{tikzpicture}
}
\caption{Summary of our results : welfare optimality and computational hardness for different classes of strategies}
\end{figure}
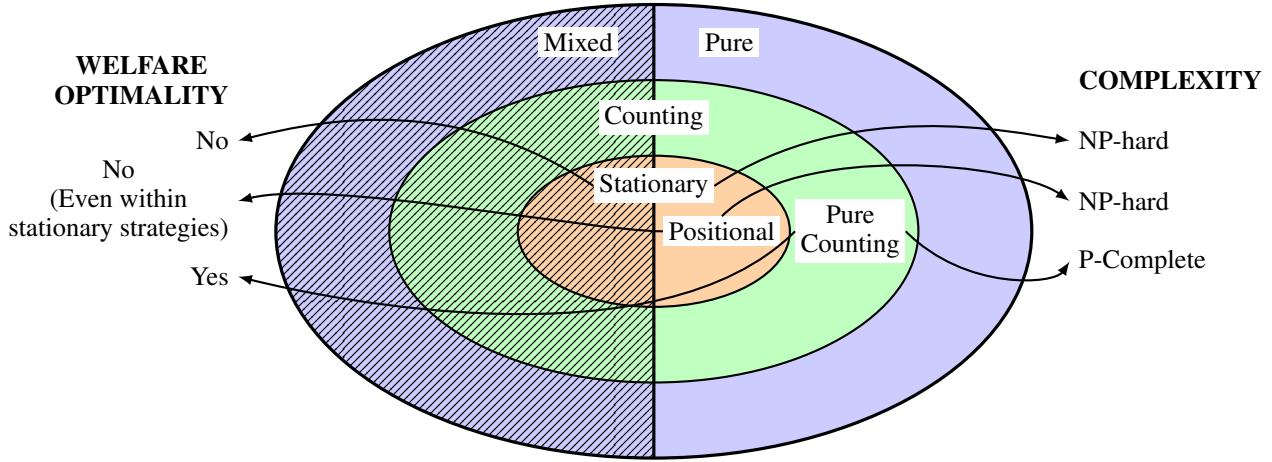

\section{Additional Details for the Main Theorems}
\label{app:theory}

\subsection{Necessity of Reasonably Spaced Discount Factors}
\label{sec:example-spaced}

Here we give an example illustrating the necessity of~\Cref{assum-spaced}. 
When discount factors are not reasonably spaced, the algorithm remains correct, but the unravelling depth can grow exponentially.
Consider
\[
\lambda_i = \frac{n}{2n-1}
\qquad\text{and}\qquad
\lambda_{i+1} = \frac{n+1}{2n+1}.
\]
The denominator $\ln(\lambda_i/\lambda_{i+1})$ in the estimate for the unravelling depth in~\Cref{eq:kappa-calculate} is approximately
$\frac{1}{(\lambda_i/\lambda_{i+1})-1}$ for large $n$, which equals $2n^2+n-1$.
If $n$ is given in binary, this quantity is exponential in the input size.

The following family of MDPs realises this behaviour.

\begin{figure}[h]
\centering
\scalebox{0.7}{
\begin{tikzpicture}[
    every node/.append style={box state/.append style={fill=safecellcolor}},
    every state/.append style={circle, fill=yellow!10, draw, minimum size=11mm},
    node distance=3cm,
    initial where=left,
    initial text={},
    prob/.style={draw,fill,shape=circle,minimum size=2mm,inner sep=0mm}
]

\node[state, initial] (q0) {$P_0$};
\node[prob, right=1.5cm of q0] (p) {};
\node[state, right=1.5cm of p] (q1) {$S_1$};
\node[state, right=2cm of q1] (q2) {$S_2$};

\path[-] (q0) edge node[above] {$0,0$} (p);

\path[->]
    (p) edge node[below] {$\frac{1}{2}$} (q1)
          edge [bend left=45] node [below] {$\frac{1}{2}$} (q0)
;

\path[->] (q1) edge [loop above] node {$1,0$} ()
           (q1) edge node [above] {$2,2$} (q2)
           ;

\path[->] (q2) edge [loop above] node {$0,2$} ();

\end{tikzpicture}
}
\end{figure} 

State $P_0$ is a probabilistic state that ensures that $S_1$ can be reached after histories of arbitrarily large length.
State $S_1$ is the only state at which a choice is made. It offers two actions.
The first action is a self-loop, giving reward $1$ to the first principal, whose discount factor is
$\lambda_1=\frac{n}{2n-1}$, and reward $0$ to the second principal, whose discount factor is
$\lambda_2=\frac{n+1}{2n+1}$.
The second action moves to the sink state $S_2$, giving reward $2$ to both principals on the transition;
at $S_2$, the rewards are $0$ and $2$ for the first and second principal, respectively.

If the self-loop is taken forever, the first principal receives payoff
$2+\frac{1}{n-1}$, while the second principal receives payoff $0$.
If the strategy moves to $S_2$, the first principal receives payoff $2$, while the second principal receives payoff
$4+\frac{2}{n}$.
Thus, staying in $S_1$ is eventually the better social choice, but initially the second principal's gain from moving to $S_2$ outweighs the first principal's gain from staying in $S_1$.
The initial imbalance is by a factor of
\[
40n - 2 - \frac{2}{n},
\]
whereas the relative discount factor is
\[
\frac{\lambda_2}{\lambda_1}
= 1 - \frac{1}{2n^2+n}.
\]
Since this ratio is extremely close to $1$, it takes exponentially many steps, in the binary encoding of $n$, before the initial imbalance is discounted below $1$.

For example, for $n=10$ we obtain $\kappa=762$, for $n=100$ we obtain
$\kappa=120324$, and for $n=1000$ the unravelling depth grows to
$\kappa=165953878$.
This shows that without the reasonable-spacing assumption, the algorithm may require exponentially many unravelling steps, even though its correctness is unaffected.

\subsection{Omitted Proof from Section~\ref{sec:pos-hard}}
\label{sec:proof-full}

Here we give the complete proof of Theorem~\ref{pos-NP}.
Let $\phi$ be a 3-SAT instance over variables
$X = \set{x_1, x_2, \ldots, x_n}$, with clauses
$C = \set{C_1, C_2, \ldots, C_m}$.

We construct an asymmetrically-discounted MDP $\mdp_\phi$ with two principals.
Let
\[
L = X \cup \set{\neg x : x \in X}
\]
be the set of literals over $X$, and let
\[
C' = \set{C'_1, C'_2, \ldots, C'_m}
\]
be a primed copy of the clause set.

We construct an asymmetrically-discounted MDP
$\mdp_\phi = (S, A, T, R, s_0)$ with:
\begin{itemize}
    \item the set of states
    \[
    S = \set{s_0, \top, \bot} \cup C \cup C^{\prime}
        \cup \set{V_x : x \in X} \cup L;
    \]

    \item the set of actions
    \[
    A = \set{\mathtt{down}} \cup \set{a_{\ell} : \ell \in L}
        \cup \set{a_\top, a_\bot} \cup \set{\mathtt{stay}};
    \]

    \item the transition function $T$, defined as follows:
    \begin{itemize}
        \item $T(s_0, \mathtt{down})(s') = \frac{1}{n+m}$ for every
        $s' \in \set{V_x : x \in X} \cup C$, and
        $T(s_0, \mathtt{down})(s') = 0$ for all other states $s'$;

        \item $T(C_i, \mathtt{down})(C'_i) = 1$ for each clause
        $C_i$, where $i \in [m]$;

        \item $T(C'_i, a_{\ell})(\ell) = 1$ for each literal
        $\ell$ appearing in clause $C_i$;

        \item $T(V_x, a_x)(x) = 1$ and
        $T(V_x, a_{\neg x})(\neg x) = 1$ for each variable $x \in X$;

        \item $T(\ell, a_\top)(\top) = 1$ and
        $T(\ell, a_\bot)(\bot) = 1$ for each literal $\ell \in L$;

        \item $T(\top, \mathtt{stay})(\top) = 1$ and
        $T(\bot, \mathtt{stay})(\bot) = 1$.
    \end{itemize}

    \item two principals, Principal~$0$ and Principal~$1$, with discount
    factors $\lambda_0 = 0.54$ and $\lambda_1 = 0.4$, respectively;

    \item identical rewards for both principals, defined as follows:
    \begin{itemize}
        \item $R(\ell, a_\top) = -1$ and $R(\ell, a_\bot) = +1$
        for each literal $\ell \in L$;

        \item $R(\top, \mathtt{stay}) = +1$ and
        $R(\bot, \mathtt{stay}) = -1$;

        \item all other rewards are $0$.
    \end{itemize}
\end{itemize}

This reduction establishes NP-hardness, even with just two principals and identical or zero-sum rewards.

\begin{proof}
The MDP starts at state $s_0$. 
It then stochastically goes through one of the variable nodes or one of the clause nodes.
If the MDP goes to variable nodes, it will take 2 steps to reach a literal and 3 steps to reach either of the sinks ($\top$ or $\bot$); while if it goes to clause nodes, it will take 3 steps to get to a literal and 4 steps to get to either of the sinks ($\top$ or $\bot$). 

On the shorter paths, Principal $0$ receives a total payoff of 
\[
0 + 0 \cdot  \lambda_0 + (-1) \cdot {\lambda_0}^2+\frac{(+1)\cdot {\lambda_0^3}}{1-\lambda_0} = \frac{729}{14375}  \approx 0.0507
\]
when going to $\top$, and $-\frac{729}{14375}   \approx -0.0507$ when going to $\bot$, while Principal $1$ receives an overall payoff of $ -\frac{4}{75} = -0.05\overline{3}$ when going to $\top$, and $\frac{4}{75} = 0.05\overline{3}$ when going to $\bot$.
Short paths that reach $\top$ therefore contribute approximately $ -\frac{113}{43125} \approx -0.003 $ to the social welfare, while going to $\bot$ on short paths contributes $+\frac{113}{43125} \approx +0.003 = c_s$.

Thus, on short paths, we would ideally reach $\bot$.
for long paths, it is the other way round: long paths reaching $\top$ contribute $\frac{13049}{2156250} \approx +0.006 = c_\ell$ to the social welfare, while long paths reaching $\bot$ contribute $ -\frac{13049}{2156250} \approx -0.006$.

Thus, a strategy $\sigma$ can achieve a social welfare of $\frac{m \cdot c_\ell + n \cdot c_s}{m+n}$ if, and only if, all reachable short paths go to $\bot$, while all reachable long paths go to $\top$; otherwise the social welfare of a strategy is strictly smaller.

Such a strategy can be taken from a solution to the 3SAT problem (go to $\top$ from true and to bot from false literals and go true literals on the long and false literals on the short paths).
Likewise, any strategy that produces this social welfare has at most one of the literals for every variable true, and at least one literal for each clause.
Thus, when selecting as true the literals that turn to $\top$ then provides a solution to the SAT problem.
This closes the proof for the case of two principals with the same payoff. For two-principal zero-sum games we use the same reductions, but different discount factors and payoffs.
\end{proof}

\section{Additional Experimental Results}
\label{app:experiments}

\subsection{Enlarged Runtime Plots}
\label{experiments}
Figure~\ref{fig:combined_time_large} provides enlarged versions of the runtime plots from Section~\ref{sec:experiment}. 
These plots show how the running time varies with the number of states, the number of principals, and the spacing of discount factors.

\begin{figure*}[t!]
\centering

\begin{subfigure}{0.48\textwidth}
    \centering
    \includegraphics[width=0.93\linewidth]{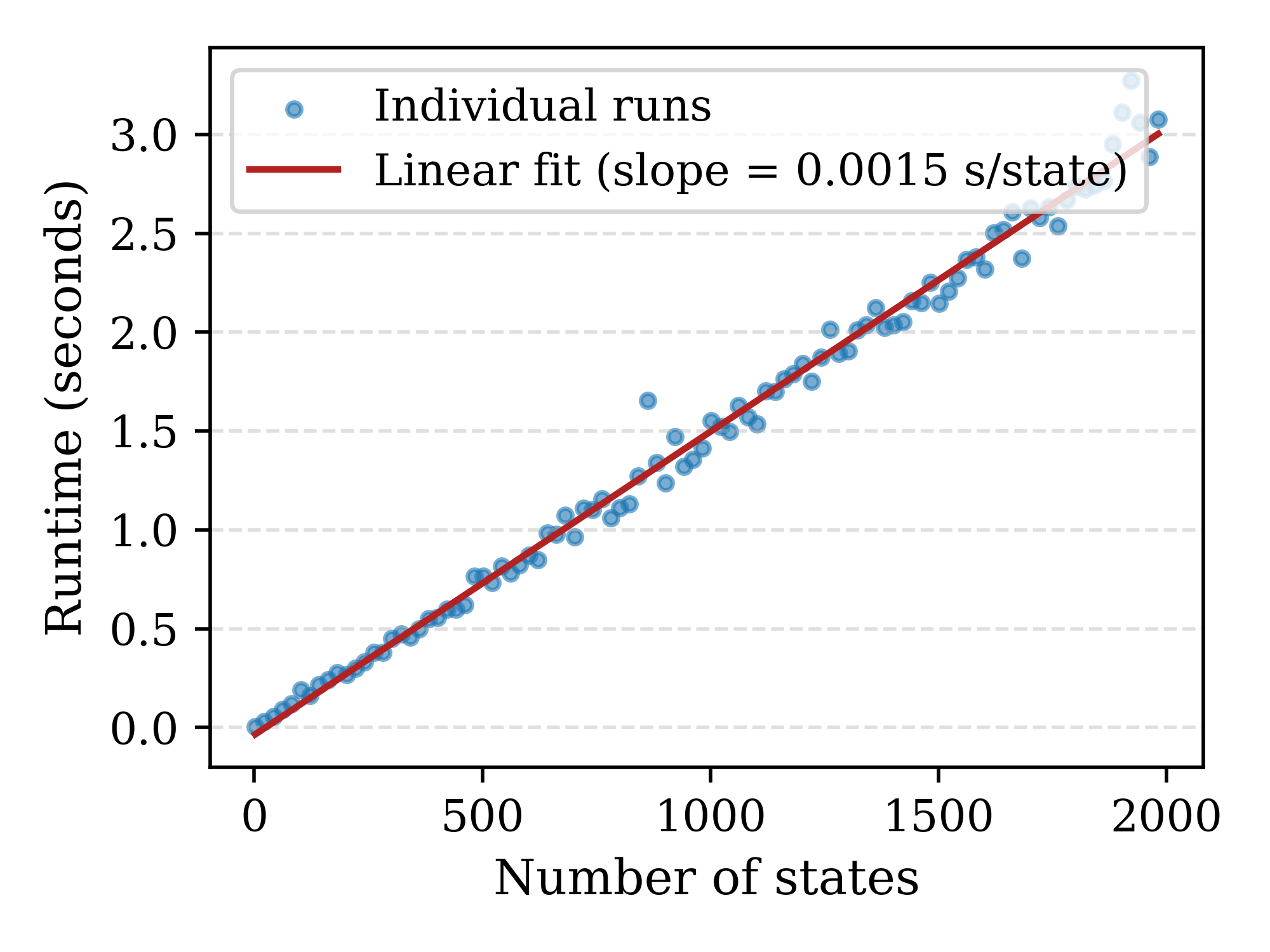}
    \caption{States growing linearly}
\end{subfigure}
\hfill
\begin{subfigure}{0.48\textwidth}
    \centering
    \includegraphics[width=0.93\linewidth]{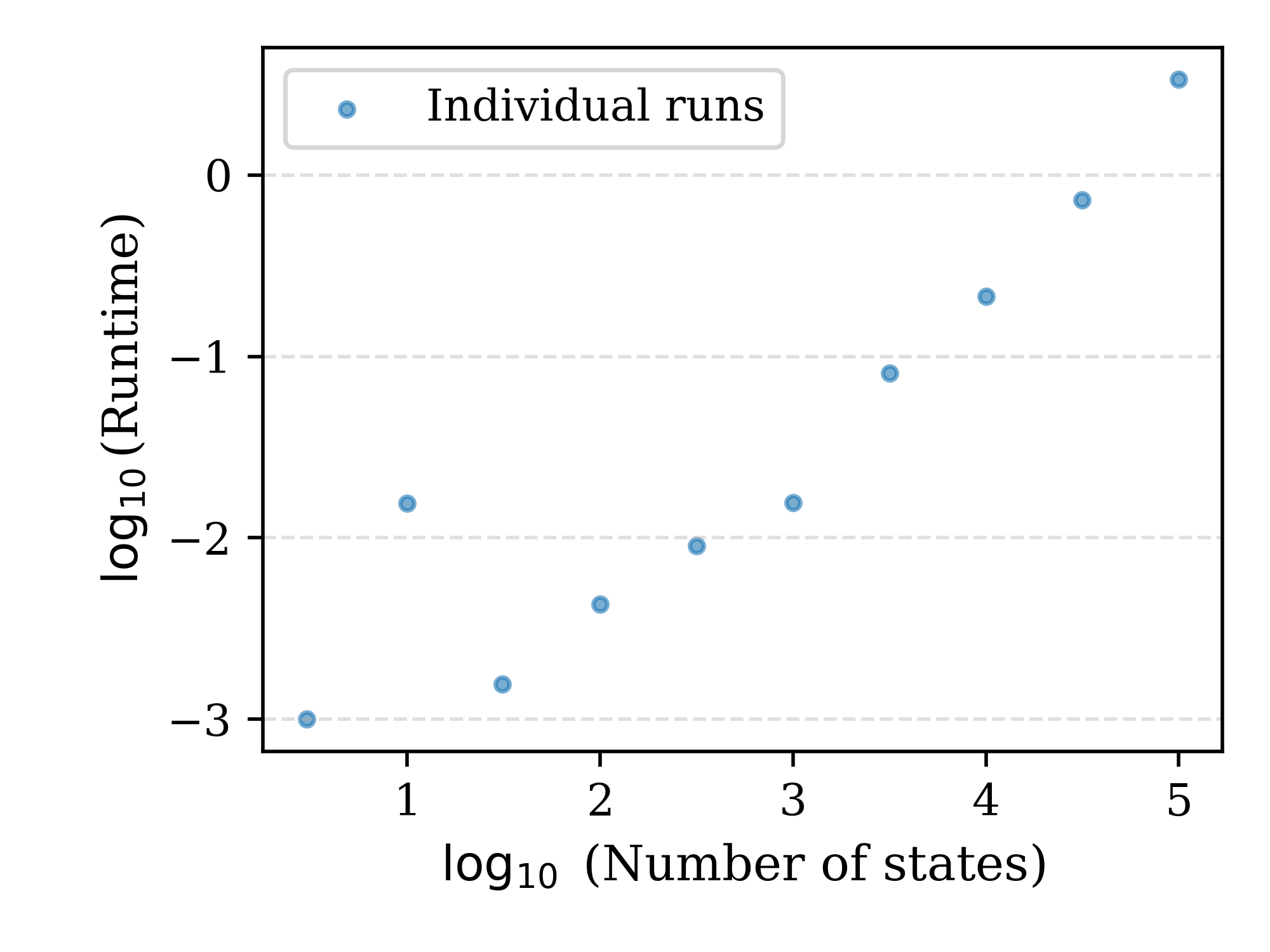}
    \caption{States scaled from $3$ to $100{,}000$}
\end{subfigure}

\vspace{0.6em}

\begin{subfigure}{0.48\textwidth}
    \centering
    \includegraphics[width=0.93\linewidth]{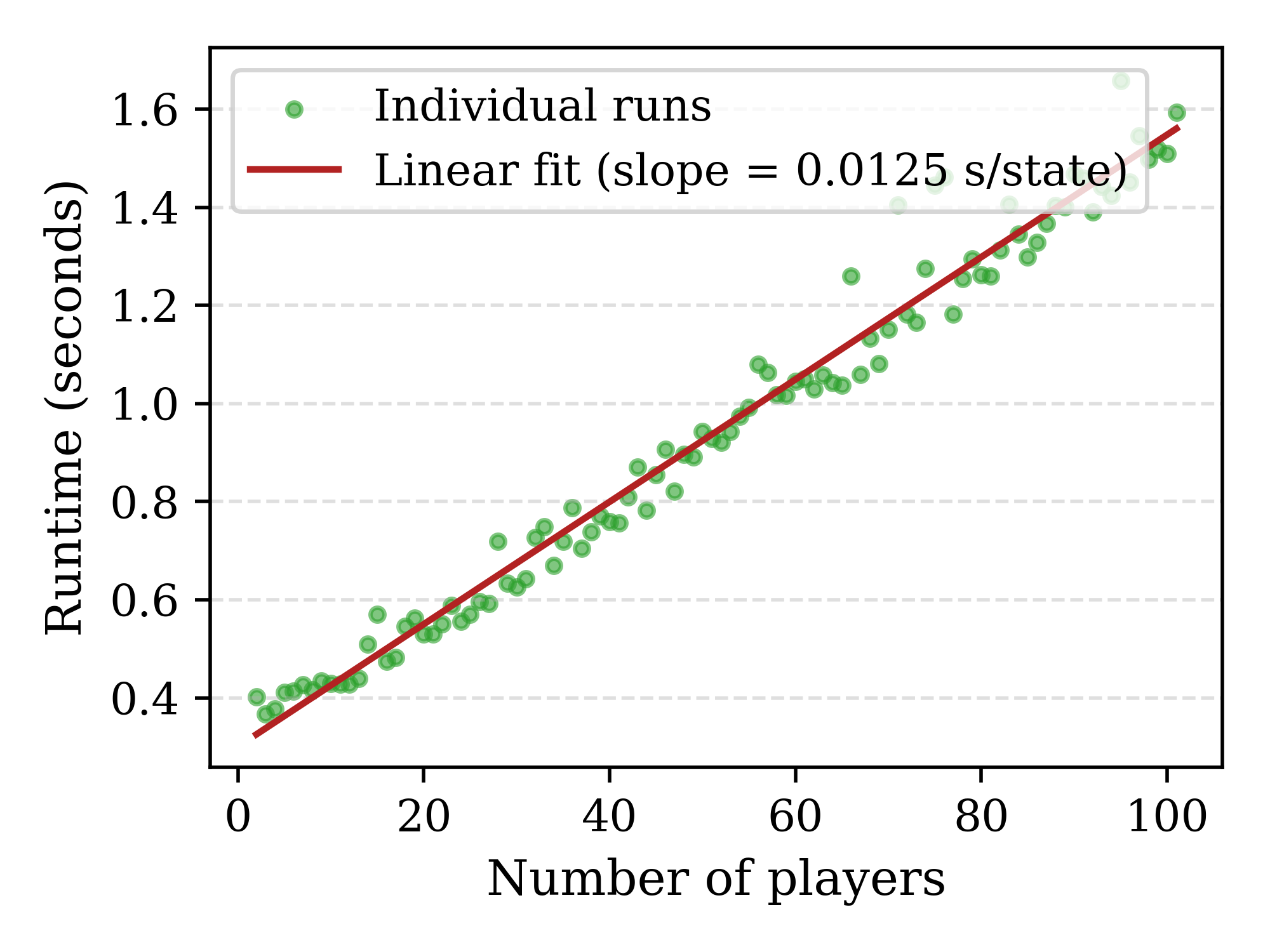}
    \caption{Principals growing linearly}
\end{subfigure}
\hfill
\begin{subfigure}{0.48\textwidth}
    \centering
    \includegraphics[width=0.93\linewidth]{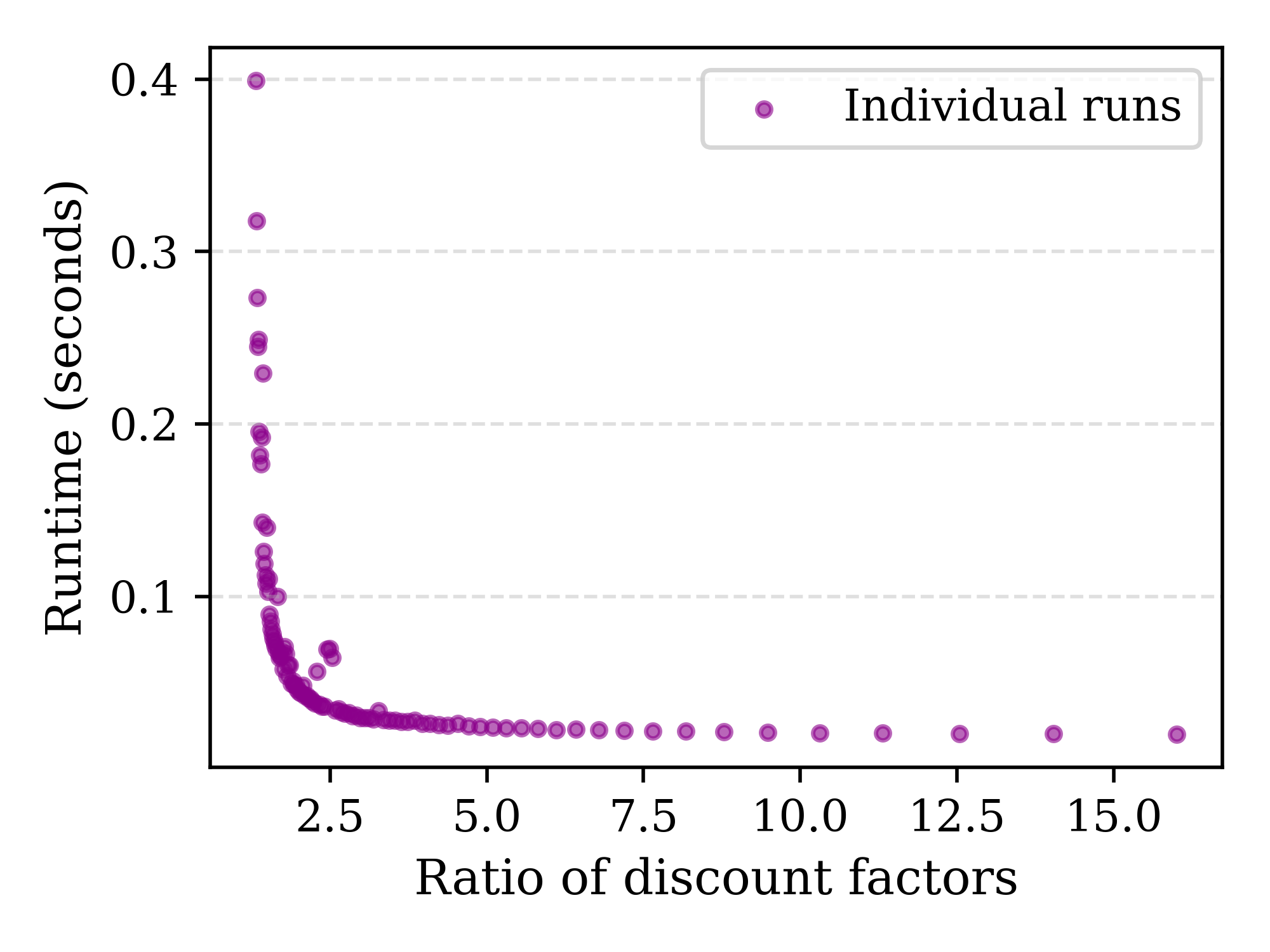}
    \caption{Ratio of discount factors}
\end{subfigure}

\caption{Computational time under varying model parameters.}
\label{fig:combined_time_large}
\end{figure*}

\subsection{Runtime With and Without MDP Solving}
\label{app:runtime-mdp-solving}
Figure~\ref{fig:large-scale-nonLog} reports the relationship between the number of states and runtime, corresponding to Figure~\ref{fig:combined_time_large}~(top-right). 
The left plot shows the total runtime, including both MDP solving via Algorithm~\ref{algo:long-term} and the welfare-optimisation algorithm. 
The right plot excludes the MDP-solving phase and reports only the runtime of our algorithm. 
MDP solving is performed using value iteration, which dominates the runtime for large instances due to slow convergence in practice.

\begin{figure*}[t!]
\centering
\begin{subfigure}{0.48\textwidth}
    \centering
    \includegraphics[width=\linewidth]{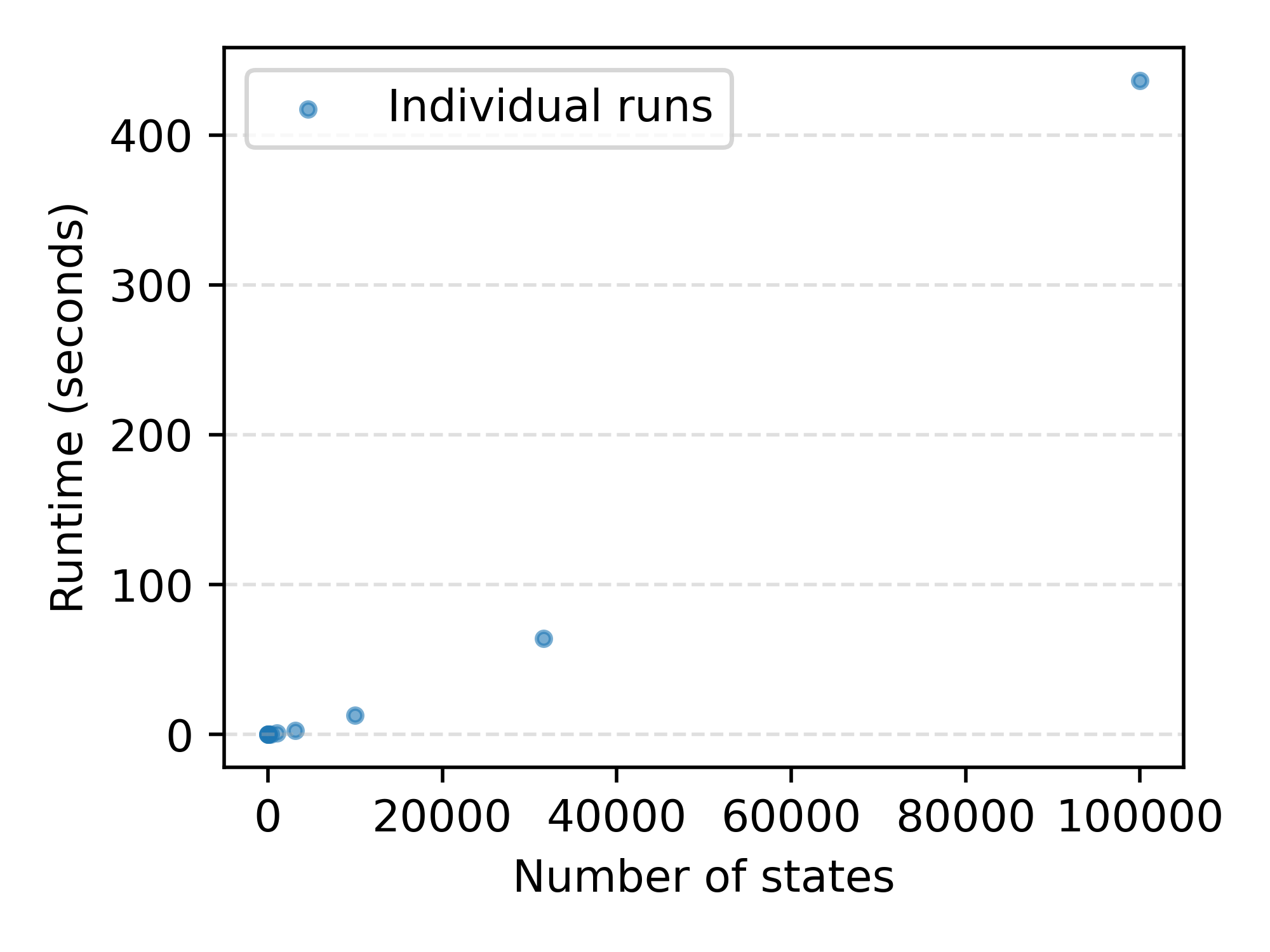}
\end{subfigure}\hfill
\begin{subfigure}{0.48\textwidth}
    \centering
    \includegraphics[width=\linewidth]{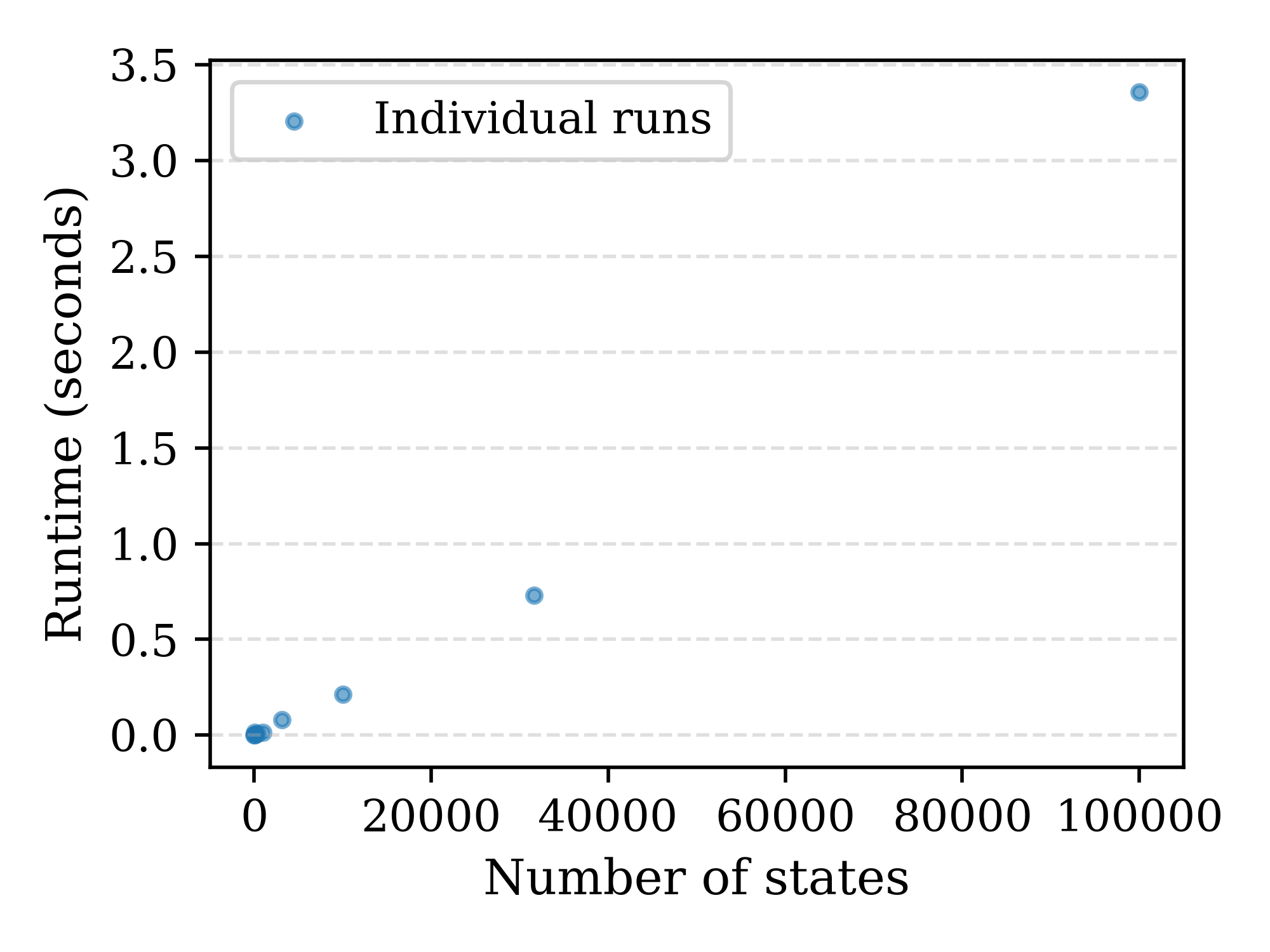}
\end{subfigure}
\caption{Runtime for evaluating social welfare with respect to the number of states; with (left) and without (right) considering MDP-solving.}
\label{fig:large-scale-nonLog}
\end{figure*}

\section{Worked Examples}

Here we provide few worked examples. We provide more details on games from main paper (Figure \ref{fig:exp-same-payoff} below in Example 1.  
We also discuss few new examples.
These examples are worked out with our implementation of the algorithm presented in Section~\ref{sec:algo}.

\begin{figure}[t!]
    \centering
    \scalebox{0.8}{
    \begin{tikzpicture}[
      every state/.append style={fill=safecellcolor},
      every node/.append style={box state/.append style={fill=safecellcolor}},
      initial where=left,
      initial text={}
    ]
      \node[state, initial] (s0) at (0,0) {$s_0$};
      \node[state] (s1) [right=3cm of s0] {$s_1$};
      \path[->]
        (s0) edge [] node[above] {$b, -1$}(s1)
             edge [loop above] node[above] {$a, 3$} (s0)
        (s1) edge [loop above] node[above] {$b, 6$} (s1) ;
    \end{tikzpicture}}
      \caption{An asymmetrically-discounted MDP with two principals with discount factors: $\lambda_0=\frac{2}{3}, \lambda_1=\frac{1}{3}$.}
    \label{fig:paper}
\end{figure}
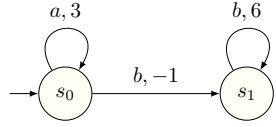

\paragraph{Example 1.} Consider the MDP shown in the introduction (Figure~\ref{fig:exp-same-payoff}). To achieve the optimal social welfare, the MDP initially deviates from the long-term strategy computed by Algorithm~\ref{algo:long-term} for two steps.

\begin{enumerate}
    \item \textbf{Long-term strategy (Algorithm~\ref{algo:long-term}).} The long-term strategy derived from the perspective of the most patient principal (Principal~0), refined by lexicographic tie-breaking, is:
    \[
    \pi_{\infty} = \{s_0 \mapsto b,\; s_1 \mapsto b\}.
    \]

    \item \textbf{Long-term values.} The value functions computed under $\pi_{\infty}$ are:
    \[
    V_0(s_0) = 11 \quad \text{and} \quad V_1(s_0) = 2.
    \]
    The baseline long-term component of social welfare is therefore $V_0(s_0) + V_1(s_0) = 13$.

    \item \textbf{Deviation phase.} For action $a$ at state $s_0$, the payoff differences $\Delta_j(s_0, a, i)$ are strictly positive for $j < \kappa = 2$. Thus, the optimal strategy initially deviates from $\pi_{\infty}$ for two steps to accumulate additional welfare before switching permanently to $\pi_{\infty}$.

    \item \textbf{Welfare-optimal strategy.} The resulting pure, finite-memory strategy that maximises social welfare is:
    \[
    \sigma = \langle s_0 \mapsto a,\; s_0 \mapsto a,\; s_0 \mapsto b,\; (s_1 \mapsto b)^\omega\rangle.
    \]
    This strategy can be implemented using a counter of depth $\kappa = 2$ followed by the stationary long-term strategy $\pi_{\infty}$.

    \item \textbf{Total social welfare.} The total social welfare under strategy $\sigma$ is the sum of the long-term values and the gain from deviation:
    \[
    \sw_{\sigma}(s_0) = V_0(s_0) + V_1(s_0) + \sum_{j=0}^{\kappa-1} \mathbb{E} \sum_{i=0}^{n-1} \Delta_j(s_j, a_j, i) \approx 11 + 2 + 1.2 = 14.2.
    \]
\end{enumerate}

\begin{figure}[t!]
    \centering
\begin{tikzpicture}[
    every state/.append style={fill=safecellcolor},
    every node/.append style={box state/.append style={fill=safecellcolor}},
    initial text={}
  ]
  \node[state, initial] (s0) at (0,0) {$s_0$};
  \node[prob state] (p0a) [right=2cm of s0, yshift=2cm] {};
  \node[prob state] (p0b) [right=2cm of s0, yshift=-2cm] {};

  \node[state] (s2) [right=4cm of s0, yshift=3cm] {$s_2$};
  \node[state] (s1) [right=4cm of s0, yshift=1cm] {$s_1$};
  \node[state] (s3) [right=4cm of s0, yshift=-2cm] {$s_3$};

  \node[prob state] (p1c) [right=2cm of s1, yshift=0cm] {};
  \node[prob state] (p1d) [right=-2cm of s1, yshift=-1cm] {};


  \path[-] (s0) edge node[left] {$a$, 2.75} (p0a);
  \path[-] (s0) edge node[left] {$b$} (p0b);

  \path[->]
    (p0a) edge node[below] {$0.5$} (s1)
          edge node[above] {$0.5$} (s2)
    (p0b) edge[bend left=45] node[left] {$0.6$} (s0)
          edge node[below] {$0.4$} (s3);

  \path[-] (s1) edge node[above] {$c$, 3} (p1c);
  \path[-] (s1) edge node[below] {~~$d$, -0.5} (p1d);

  \path[->]
    (p1c) edge[bend left=45] node[below] {$0.3$} (s1)
          edge node[right] {$0.4$} (s2)
          edge[bend left=45] node[left] {$0.3$} (s3)
          
    (p1d) edge node[above] {$0.5$} (s0)
          edge node[right] {$0.5$} (s3);

  \path[->]
    (s2) edge[loop above] node[above] {$e$, 1} (s2)
    (s3) edge[loop below] node[below] {$f$, 2.5} (s3);
\end{tikzpicture}

    \caption{An asymmetrically-discounted MDP with two Principals with discount factors: $\lambda_0=0.9, \lambda_1=0.3$}
    \label{fig:mdplast}
\end{figure}
\paragraph{Example 2.} Consider the MDP shown in Figure~\ref{fig:mdplast}. For every state-action pair discarded by Algorithm~\ref{algo:long-term}, the gain for Principal~1 does not compensate for the loss incurred by Principal~0; that is, the loss for Principal~0 outweighs the gain for Principal~1. Since the discount factors are reasonably spaced (with $\lambda_0 > \lambda_1$), this imbalance only worsens over time, as future gains are increasingly down-weighted. Consequently, $\kappa = 1$.

\begin{enumerate}
    \item \textbf{Long-term strategy (Algorithm~\ref{algo:long-term}).} The long-term strategy computed by Algorithm~\ref{algo:long-term} is:
    \[
    \pi_{\infty} = \{s_0 \mapsto b,\; s_1 \mapsto d,\; s_2 \mapsto e,\; s_3 \mapsto f\}.
    \]

    \item \textbf{Long-term values.} From the initial state $s_0$, the value functions are:
    \[
    V_0(s_0) \approx 21.83 \quad \text{and} \quad V_1(s_0) \approx 1.79.
    \]

    \item \textbf{Deviation phase.} In this MDP, $\kappa = 1$; all potential deviations are immediately suboptimal. Thus, the welfare-optimal strategy does not deviate from $\pi_{\infty}$.

    \item \textbf{Welfare-optimal strategy.} The strategy that maximises social welfare is:
    \[
    \pi_{\infty} = \{s_0 \mapsto b,\; s_1 \mapsto d,\; s_2 \mapsto e,\; s_3 \mapsto f\}.
    \]

    \item \textbf{Total social welfare.} Since no deviation occurs, the total social welfare is:
    \[
    \sw_{\sigma}(s_0) = V_0(s_0) + V_1(s_0) \approx 23.62.
    \]
\end{enumerate}

\begin{figure}[t!]
    \centering
    \begin{tikzpicture}[
      every state/.append style={fill=safecellcolor},
      every node/.append style={box state/.append style={fill=safecellcolor}},
      initial where=left,
      initial text={}
    ]
      \node[state, initial] (s0) at (0,0) {$s_0$};
      \node[state] (s1) [right=3cm of s0] {$s_1$};
      \node[state] (s2) [right=3cm of s1] {$s_2$};
      \node[state] (s3) [below=2cm of s2] {$s_3$};
      \node[state] (s4) [below=2cm of s0] {$s_4$};
      \node[state] (s5) [below=2cm of s1] {$s_5$};

      \path[->]
        (s0) edge [] node[left] {$a, 1$} (s4)
             edge [] node[above] {$b, 0$} (s1)
        (s1) edge [loop above] node[above] {$c, 10$} (s1)
             edge [] node[above] {$d, 0$} (s2)
        (s2) edge [loop right] node[right] {$e, 7$} (s2)
             edge [] node[right] {$f, 5$} (s3)
        (s3) edge [loop right] node[right] {$g, 11$} (s3)
        (s4) edge [] node[below] {$h, 0$} (s5)
             edge [] node[right] {$j, -2$} (s1)
        (s5) edge [] node[below] {$k, 0$} (s3);
    \end{tikzpicture}
    \caption{An asymmetrically-discounted MDP with two principals with discount factors: $\lambda_0=0.99, \lambda_1=0.01$}
\label{rectangle-kappa0}
\end{figure}
\paragraph{Example 3.} Consider the MDP shown in Figure~\ref{rectangle-kappa0}. In this example, the MDP always follows the long-term strategy computed by Algorithm~\ref{algo:long-term}. This means that deviating to any discarded state-action pair at the initial state yields a gain for Principal~1 that does not compensate for the corresponding loss incurred by Principal~0.

\begin{enumerate}
    \item \textbf{Long-term strategy (Algorithm~\ref{algo:long-term}).} The long-term strategy computed by Algorithm~\ref{algo:long-term} is:
    \[
    \pi_{\infty} = \{s_0 \mapsto b,\; s_1 \mapsto d,\; s_2 \mapsto f,\; s_3 \mapsto g,\; s_4 \mapsto h,\; s_5 \mapsto k\}.
    \]

    \item \textbf{Long-term values.} The value functions under $\pi_{\infty}$ from the initial state $s_0$ are:
    \[
    V_0(s_0) = 1072.2294 \quad \text{and} \quad V_1(s_0) = 1.0000\overline{1}.
    \]
    Hence, the baseline component of social welfare is already very close to the optimum.

    \item \textbf{Deviation phase.} In this case, $\kappa = 1$; every removed edge reachable from $s_0$ has a non-positive payoff difference at the moment it can first be taken. Therefore, the optimal finite-memory strategy does not deviate from $\pi_{\infty}$, as no gain in social welfare can be achieved through deviation.

    \item \textbf{Welfare-optimal strategy.} The strategy that maximises social welfare is simply the long-term strategy:
    \[
    \sigma = \langle s_0 \mapsto b,\; s_1 \mapsto d,\; s_2 \mapsto f,\; (s_3 \mapsto g)^\omega \rangle.
    \]
    This is a pure stationary strategy induced directly by $\pi_{\infty}$.

    \item \textbf{Total social welfare.} Since no deviation occurs, the total social welfare is simply the sum of the long-term values:
    \[
    \sw_{\sigma}(s_0) = V_0(s_0) + V_1(s_0) \approx 1073.23.
    \]
\end{enumerate}

\begin{figure}[t!]
    \centering
    {
    \begin{tikzpicture}[
      every state/.append style={fill=safecellcolor},
      every node/.append style={box state/.append style={fill=safecellcolor}},
      initial where=left,
      initial text={}
    ]
      \node[state, initial] (s0) at (0,0) {$s_0$};
      \node[state] (s1) [right=3cm of s0] {$s_1$};
      \node[state] (s2) [right=3cm of s1] {$s_2$};
      \node[state] (s3) [below=2cm of s2] {$s_3$};
      \node[state] (s4) [below=2cm of s0] {$s_4$};
      \node[state] (s5) [below left=2cm and 1cm of s1] {$s_5$};
      \node[state] (s6) [below left=2cm and 2cm of s2] {$s_6$};

      \path[->]
        (s0) edge [] node[left] {$a, 10.5$} (s4)
             edge [] node[above] {$b, 0$} (s1)
        (s1) edge [loop above] node[above] {$c, 10$} (s1)
             edge [] node[above] {$d, 0$} (s2)
        (s2) edge [loop right] node[right] {$e, 7$} (s2)
             edge [] node[right] {$f, 5$} (s3)
        (s3) edge [loop right] node[right] {$g, 20$} (s3)
        (s4) edge [] node[below] {$h, 0$} (s5)
             edge [] node[above] {$j, -2$} (s1)
        (s5) edge [] node[below] {$k, 0$} (s6)
        (s6) edge [] node[below] {$m, 0$} (s3);
        
    \end{tikzpicture}
    }
    \caption{An asymmetrically-discounted MDP with two principals with discount factors: $\lambda_0=0.88, \lambda_1=0.15$}
    \label{ex-copy}
\end{figure}
\paragraph{Example 4.} Consider the MDP in Figure~\ref{ex-copy}. This is almost structurally similar to the one in Figure~\ref{rectangle-kappa0} (with addition of state $s_6$), and also differs in both the discount factors and reward values for certain actions. The key distinction is that, while in Figure~\ref{rectangle-kappa0} the welfare-optimal strategy coincides with the long-term strategy, in Figure~\ref{ex-copy} the MDP initially deviates. Specifically, at state~$s_0$, the optimal strategy chooses action~$a$ instead of the long-term choice~$b$. After this single-step deviation, from state~$s_4$ onward, the MDP follows the long-term strategy, starting with action~$j$.

\begin{enumerate}
    \item \textbf{Long-term strategy (Algorithm~\ref{algo:long-term}).} The long-term strategy computed by Algorithm~\ref{algo:long-term} is:
    \[
    \pi_{\infty} = \{s_0 \mapsto b,\; s_1 \mapsto d,\; s_2 \mapsto f,\; s_3 \mapsto g,\; s_4 \mapsto j,\; s_5 \mapsto k,\; s_6 \mapsto m\}.
    \]

    \item \textbf{Long-term values.} From the initial state $s_0$, the value functions are:
    \[
    V_0(s_0) \approx 117.45 \quad \text{and} \quad V_1(s_0) \approx 0.19.
    \]

    \item \textbf{Deviation phase.} In this case, $\kappa = 2$; action~$a$ at $s_0$ has a strictly positive payoff difference. The welfare-optimal strategy deviates at step~$0$ by taking action~$a$, then selects action~$j$ at $s_4$, and follows $\pi_{\infty}$ from that point forward.

    \item \textbf{Welfare-optimal strategy.} The strategy that maximises social welfare is:
    \[
    \sigma = \langle s_0 \mapsto a,\; s_4 \mapsto j,\; s_1 \mapsto d,\; s_2 \mapsto f,\; (s_3 \mapsto g)^\omega \rangle.
    \]

    \item \textbf{Total social welfare.} The total social welfare under $\sigma$ is the sum of long-term values and the gain from deviation:
    \[
    \sw_{\sigma}(s_0) \approx 117.45 + 0.19 + 4.69 = 122.33.
    \]
\end{enumerate}

Examples~2 and~3 demonstrate cases where Algorithm~\ref{algo:long-term}
already yields a welfare-optimal strategy, and no finite unrolling of the MDP
into a DAG is necessary. In contrast, Examples~1 and~4 show that limited
early deviations from the long-term strategy can improve social welfare and
must be explicitly incorporated.


\end{document}